\RequirePackage{tikz}
\documentclass[sn-mathphys]{sn-jnl}

\jyear{2021}%
\usepackage{physics}
\usepackage{amsthm}
\usepackage{amsmath}
\usepackage{tikz}
\usepgflibrary{shapes.misc}
\usetikzlibrary{arrows,positioning,fit,backgrounds}
\usetikzlibrary{shapes.misc} 
\usetikzlibrary{calc}

\newtheorem{defn}{Definition}
\newtheorem{lem}{Lemma}
\newtheorem{thm}{Theorem}
\newtheorem{prop}[lem]{Proposition}

\newtheorem{rem}{Remark}
\newtheorem{assump}{Assumption}
\newtheorem{example}{Example}

\theoremstyle{thmstyletwo}%

\theoremstyle{thmstylethree}%

\DeclareMathOperator{\vol}{vol}

\DeclareMathOperator{\conv}{conv}

\DeclareMathOperator{\supp}{supp}

\newcommand{\uy}{\underline{y}}
\newcommand{\uz}{\underline{z}}
\newcommand{\uY}{\underline{Y}}
\newcommand{\uZ}{\underline{Z}}

\begin{document}

\title[Minoration via Mixed Volumes and Cover's Problem for General Channels]{Minoration via Mixed Volumes and Cover's Problem for General Channels}


\author*[1]{\fnm{Jingbo} \sur{Liu}}\email{jingbol@illinois.edu}

\affil*[1]{\orgdiv{Department of Statistics}, \orgname{University of Illinois, Urbana-Champaign}, \orgaddress{\street{725 S. Wright St, 101 Illini Hall}, \city{Champaign}, \postcode{61820}, \state{IL}, \country{USA}}}


\abstract{We give a complete solution to an open problem of Thomas Cover in 1987 about the capacity of a relay channel in the general discrete memoryless setting without any additional assumptions.
The key step in our approach is to lower bound a certain soft-max of a stochastic process by convex geometry methods, which is based on two ideas:
First, the soft-max is lower bounded in terms of the supremum of another process, by approximating a convex set with a polytope with bounded number of vertices.
Second, using a result of Pajor, the supremum of the process is lower bounded in terms of packing numbers by means of mixed-volume inequalities (Minkowski's first inequality).}




\maketitle

\section{Introduction}
The supremum of stochastic processes is a major topic in high-dimensional probability \cite{boucheron2004concentration,vershynin2018high,van2014probability}, and has found applications in many fields including asymptotic statistics, learning theories, and Banach space theories \cite{dudley,geer,geer2,ledoux2013probability}. 
Information theory, since the 1970s, has long been a major testbed of high-dimensional probability tools such as hypercontractivity and concentration of measure \cite{ahlswede1976spreading,ahlswede_bounds_cond1976}.

In this paper we give a complete solution to an open problem in information theory by Thomas Cover in 1987 about the capacity of a relay channel (see \cite{cover}, \cite[Section~1]{WuBarnesOzgur}), by unveiling a link between information inequalities and the supremum of stochastic processes.
Background on the relay channel and Cover's problem will be given in  Section~\ref{sec_cover} ahead. 
Our solution to Cover's problem is given in Section~\ref{sec_general}.
Previous solutions for the Gaussian  \cite{WuBarnesOzgur}\cite{bai} and the binary symmetric \cite{BarnesWuOzgur-BSC} cases were based on rearrangement inequalities or constrained optimal transport in 2-norm,
which are specialized to those channel distributions.
Other approaches based on measure  concentration \cite{Zhang,WuOzgur,WuOzgurXie,WuOzgur-general}
 and reverse hypercontractivity \cite{liuozgur} (building on a method of \cite{lvv2018,liu_thesis,liu18}) apply for general channels but are not strong enough in the regime of interest for Cover's problem.

In the remainder of this section, we will explain some high level ideas on how these concepts are connected: We first view information theoretic quantities as a certain ``soft-max'' of information, which is further reduced to a hard-max (supremum) by approximation of convex polytopes. 
Then, lower bounding the supremum (minoration) is achieved by a result of Pajor using mixed volume inequalities.

\subsection{From Divergence to Soft-max}
Recall that the relative entropy between two probability measures $P$ and $Q$ equals\footnote{Unless otherwise noted, the bases in this paper are natural.} $D(P\|Q):=\int \log \frac{dP}{dQ}dP=\log\mathbb{E}[\exp(f(Z))]-\mathbb{E}[f(Z)]$ where $f:=-\log\frac{dP}{dQ}$ and $Z\sim P$.
Note that $\log\mathbb{E}[\exp(f(Z))]-\mathbb{E}[f(Z)]=\log\mathbb{E}[\exp(f(Z)-\mathbb{E}[f(Z)])]$ can be understood as a ``soft-max''
\footnote{The name ``soft-max'' is justified by the fact that if $X$ is a random variable equiprobably distributed on a finite set $\mathcal{X}\subseteq \mathbb{R}$, then $\max_{x\in\mathcal{X}}x-\log\abs{\mathcal{X}}\le \log\mathbb{E}[\exp(X)]\le \max_{x\in\mathcal{X}}x$.} 
of $f(Z)-\mathbb{E}[f(Z)]$; let us denote it as ${\rm smax}_{z\in\mathcal{A}}(f(z)-\mathbb{E}[f(Z)])$ in the case where $Z$ follows the equiprobable distribution on a set $\mathcal{A}$.

For our purpose we actually need to consider a conditional relative entropy, which is the average over some relative entropy over another random variable $Y$.
Given a function $K$ of discrete variables $y$ and $z$, we can always construct mappings $\phi$ and $\psi$ to a common Euclidean space so that $K(y,z)=\left<\phi(y),\psi(z)\right>$ is an inner product.
This enables us to solve a discrete problem using geometric tools.
While such ``representations'' $\phi$ and $\psi$ are not unique, it will be seen that their particular choice does not affect our bounds in the end. This idea is similar to the kernel methods used in machine learning whereby features from an arbitrary domain are mapped to a Hilbert space \cite{smale}.
By a slight abuse of notation, let us write $\phi(y)$ and $\psi(z)$ as $y$ and $z$.
It turns out that the information theoretic problem under consideration can be rephrased as follows:
\begin{align}
\textit{Upper bound }{\sf P}_l(\mathcal{A})\textit{ in terms of }\mathbb{E}[{\rm smax}_{z\in\mathcal{A}}\left<z,Y\right>]
\label{e_goal}
\end{align}
where ${\sf P}_l(\mathcal{A})$ denotes the packing number under a metric induced by a convex polytope, and $Y$ is distributed on the vertices of the convex polytope.
The packing number can be bounded in terms of the entropy of $Z$, whereas $\mathbb{E}[{\rm smax}_{z\in\mathcal{A}}\left<z,Y\right>]$ is the relative entropy.
Details about the reduction from conditional relative entropy to expected soft-max in Cover's problem will be given in Section~\ref{sec_gn}-\ref{sec_h}.

A feature of our proof which differs from standard techniques in multiuser information theory (e.g.\ \cite{el2011network}) 
is the focus on relative entropy and the capacity-achieving \emph{output} distribution (Section~\ref{sec_caod}), as opposed to optimization of mutual information over input distributions. 
The capacity-achieving output distribution is always unique for general channels, whereas the capacity-achieving input distribution is not unless further assumptions such as full-rankness are imposed.
Therefore our solution to Cover's problem does not require any additional assumption such as full-rankness.

Although not the focus of the present paper, 
the Gaussian version of Cover's problem (previously solved in \cite{WuBarnesOzgur}\cite{bai} using Gaussian tools) can also be solved using our techniques.
In fact, the geometric intuition is even more evident in the Gaussian setting, since $K(y,z)$ will then become a quadratic form in $y$ and $z$, so that the inner product structure is natural.

\subsection{From Soft-Max to Supremum}
\label{sec_soft}
We will lower bound the expected soft-max $\mathbb{E}[{\rm smax}_{z\in\mathcal{A}}\left<z,Y\right>]$ in terms of the expected supremum (hard-max) of another process, so that lower bound techniques for the latter can be utilized.
We first find a finite set $\mathcal{S}$ in $\mathcal{C}$ and show that $\mathbb{E}[{\rm smax}_{z\in\mathcal{A}}\left<z,\hat{Y}\right>]\le \mathbb{E}[{\rm smax}_{z\in\mathcal{A}}\left<z,Y\right>]$ where $\hat{Y}$ is a random variable on $\mathcal{S}$ satisfying $\mathbb{E}[\hat{Y}]=0$.
Define
\begin{align}
\mathcal{B}:=\{z\in\mathcal{A}\colon \left<z,y\right>\le \mathbb{E}[\left<Z,y\right>]
+{\rm smax}_{z\in\mathcal{A}}\left<z,y\right>
+\log(2\abs{\mathcal{S}}),\,
\forall y\in\mathcal{S}\}
\end{align}
where $Z$ is equiprobable on $\mathcal{A}$.
Using the union bound and Markov's inequality we obtain
\begin{align}
\mathbb{P}[Z\in \mathcal{B}]\ge \frac1{2},
\label{eb}
\end{align}
and moreover,
\begin{align}
\mathbb{E}\left[\sup_{z\in\mathcal{B}}\left<z,\hat{Y}\right>\right]
\le 
\mathbb{E}\left[{\rm smax}_{z\in\mathcal{A}}\left<z,\hat{Y}\right>\right]
+\log(2\abs{\mathcal{S}}).
\label{e_intro2}
\end{align}
We will later use a minoration inequality of Pajor to lower bound the left side of  \eqref{e_intro2} by the packing number of $\mathcal{B}$, which is in turn lower bounded by the packing number of $\mathcal{A}$ in view of \eqref{eb}. 
This will fulfill our goal in \eqref{e_goal}.

We should not let $\mathcal{S}$ be too large to control the second term on the right side of \eqref{e_intro2}.
At the same time, in order to translate between the convex distances associated with $\mathcal{C}$ and $\mathcal{S}$, we cannot choose $\mathcal{S}$ too small.
For this, the mathematical core is the following question (Section~\ref{sec_sampling}):
\begin{quote}
\it Find Banach space $B$ for which the unit ball has at most $m$ vertices and the Banach-Mazur distance $d(B,\ell_{\infty}^n)$ is small.
\end{quote}
Here, $\ell_{\infty}^n$ denotes the $n$-dimensional space under the $\ell_{\infty}$ norm.
It turns out that $m=\exp(O(\frac{n}{a^2}\log a))$ is sufficient for $d(B,\ell_{\infty}^n)=a>1$ independent of $n$.
This can be seen from a general result on thrifty approximations of general convex bodies in \cite{barvinok2014thrifty} based on John's decomposition;
however we will provide a simple and explicit construction using Hadamard matrices which achieves this bound for $\ell_{\infty}^n$.
For the Gaussian version of Cover's problem,
we need similar results for $d(B,\ell_2^n)$ instead, which will be discussed in Section~\ref{sec_bm}.

\subsection{Lower Bounding the Supremum}
Robust tools exist for upper bounds on the supremum of rather general processes, 
whereas the lower bound (also called Sudakov's minoration; see e.g.\ \cite{ledoux2013probability}) is somewhat more delicate.
Currently, the most popular proof of Sudakov's inequality based on Gaussian comparison \cite{ledoux2013probability}\cite{chatterjee2005error} is short but highly Gaussian-specific.
Minoration inequalities for Rademacher processes, in contrast to the Gaussian processes, are often dimension dependent \cite[Chapter~4]{ledoux2013probability}\cite{pajor}.
Both the Gaussian and the Rademacher processes concern the inner product with vector with i.i.d.\ entries,
which is called canonical processes in \cite{tal_canonical}.
Minoration inequalities for other canonical processes were previously obtained, for example, in the case of distributions of the form $c\exp(-\abs{x}^{\alpha})$ \cite[Section~3.3]{ledoux2013probability}\cite{tal_canonical} or 
log-concave distributions \cite{latala2014sudakov}.

For our purpose, however, extension beyond canonical processes is necessary to control the left side of 
\eqref{e_intro2}, since $\hat{Y}$ does not follow a product distribution.
A classical result in Pajor's thesis \cite[eq 2-7] {pajort} (see also Section~\ref{sec_glemma}) via mixed volume inequalities states that
\begin{align}
{\sf P}_l(\mathcal{B})\le 
\left(1+\tfrac{2\mathbb{E}[\sup_{z\in\mathcal{B}}\left<z,\hat{Y}\right>]}{l}\right)^N
\label{e_intro1}
\end{align}
for any $l>0$,
where ${\sf P}_l(\mathcal{B})$ denotes the packing number under the metric induced by the same convex body which defines the distribution of $\hat{Y}$, and $N$ is the dimension of $z$.
Pajor applied this result for the cases of 
Gaussian and Rademacher averages and for the studies of entropy estimates in his thesis,
although modern methods have become more popular on these topics.
Though not obvious, \eqref{e_intro1} can recover the standard Sudakov's (Gaussian inequality) by choosing $\hat{Y}$ to be uniform on a sphere and applying the Johnson-Lindenstrauss lemma; see \cite{mendelson2019generalized}.
The work of 
\cite{mendelson2019generalized} was motivated by the question of extending minoration to general (non-product) log-concave measures.
Our work, on the other hand, uses \eqref{e_intro1} for the purpose of lower bounding the expected soft-max by combing the argument in Section~\ref{sec_soft}.

\section{Preliminary}
\subsection{Sudakov's Minoration}
Let us recall the basic definition of packing in metric spaces and Sudakov's inequality for the Gaussian process.
For more information, the reader may refer to \cite{boucheron2004concentration,van2014probability,ledoux2013probability}.
\begin{defn}
Let $(T,d)$ be a metric space. 
For $l>0$ and $\mathcal{A}\subseteq T$,
we say that $\mathcal{A}$ is an $l$-packing if 
$d(x,y)> l$, for all $x,y\in \mathcal{A}$, $x\neq y$.
The $l$-packing number, ${\sf P}_l(T)$, is defined as the maximum cardinality of an $l$-packing.
\end{defn}

Let $\{X_t\}_{t\in T}$ be a (centered) Gaussian process and equip $T$ with the natural metric
\begin{align}
d(t,s):=\mathbb{E}^{\frac1{2}}[\abs{X_t-X_s}^2].
\end{align}
Sudakov's inequality states that
\begin{align}
\mathbb{E}[\sup_{t\in T}X_t]
\ge c\sup_{l>0}l\sqrt{\log {\sf P}_l(T)}
\label{e_sudakov}
\end{align}
for some universal constant $c>0$.
The popular presentation of Sudakov's inequality replaces the packing number in  \eqref{e_sudakov} with the covering number even though the packing number bound follows directly from the proof,
presumably for convenient comparisons with upper bounds on the supremum (e.g.\ Dudley's integral) where covering numbers arise more naturally.
The packing number and the covering number are equivalent up to a factor of 2 in $l$ (see e.g.\ \cite{van2014probability}) so the choice is not a critical issue.

\subsection{Mixed Volumes}
The Alexandrov-Fenchel inequality is a fundamental result in convex geometry that relates the mixed volumes of convex bodies.
Many geometric inequalities, such as Minkowski's first inequality,
the Brunn-Minkowski inequality,
and the isoperimetric inequalities 
can all be regarded as its special cases
\cite{gardner}.

Given convex bodies $\mathcal{C}_1$,\dots,$\mathcal{C}_r$ in $\mathbb{R}^N$,
a fundamental fact from H.~Minkowski's theory is that the volume of the Minkowski sum
\begin{align}
\vol(\lambda_1\mathcal{C}_1+\dots+\lambda_r\mathcal{C}_r)
=\sum_{j_1,\dots,j_N=1}^r
V(\mathcal{C}_{j_1},\dots,\mathcal{C}_{j_N})
\lambda_{j_1}\dots\lambda_{j_N}
\label{e7}
\end{align}
can be shown to be a homogenous degree $N$ polynomial in $\lambda_1,\dots,\lambda_N\ge 0$.
The coefficients $V(\mathcal{C}_{j_1},\dots,\mathcal{C}_{j_N})$
are called mixed volumes.
Obvious properties of $V$ include symmetry in its arguments and multilinearity.
Moreover, $V(\mathcal{C}_{j_1},\dots,\mathcal{C}_{j_1})=\vol(\mathcal{C}_{j_1})$ is the volume.
The mixed volume $V(\mathcal{C}_1,\mathcal{C}_2,\dots,\mathcal{C}_2)$
relates to the the surface area of $\mathcal{C}_2$ when $\mathcal{C}_1$ is the unit ball,
and relates to the mean width of $\mathcal{C}_1$ when $\mathcal{C}_2$ is the unit ball.
Less obvious is the Alexandrov-Fenchel inequality which is the following ``log-concavity'' property
\begin{align}
V^2(\mathcal{C}_1,\mathcal{C}_2,\mathcal{C}_3,\dots,\mathcal{C}_N)
\ge V(\mathcal{C}_1,\mathcal{C}_1,\mathcal{C}_3,\dots,\mathcal{C}_N)
V(\mathcal{C}_2,\mathcal{C}_2,\mathcal{C}_3,\dots,\mathcal{C}_N).
\end{align}
In our application to Cover's problem we shall only use the following special case which is known as Minkowski's first inequality:
\begin{align}
V^{\frac{1}{N}}(\mathcal{A},\dots,\mathcal{A})
\cdot
V^{\frac{N-1}{N}}(\mathcal{B}\dots,\mathcal{B})
\le 
V(\mathcal{A},\mathcal{B},\dots,
\mathcal{B})
\end{align}
for two convex bodies $\mathcal{A}$ and $\mathcal{B}$.
This special case also follows by a limiting argument applied to the Brunn-Minkowski inequality.\footnote{Incidentally, as recounted by D.~Donoho, the analogies between information-theoretic inequalities and the Brunn-Minkowski inequality is one among the three that exemplify the ``beauty and purity'' of Cover's interest \cite{memoriam}.}
However, in general the Alexandrov-Fenchel inequality does not seem to follow from the Brunn-Minkowski inequality, and is generally considered as a deeper result with connections to diverse branches of mathematics; for recent surveys see for example \cite{schneider,shenfeld,onemore}.

\subsection{Norms Defined by Convex Bodies}
Given a convex set $\mathcal{K}$ and a vector $x$ in some Euclidean space $\mathbb{R}^N$, define 
\begin{align}
{\sf m}_{\mathcal{K}}(x):=\inf\{\lambda>0\colon
x\in \lambda \mathcal{K}\}.
\label{e_min_func}
\end{align}
Note that while convexity guarantees subadditivity, \eqref{e_min_func} defines a norm in $\mathbb{R}^N$ only if the following additional conditions hold:
\begin{itemize}
\item $\mathcal{K}$ is bounded, so that ${\sf m}_{\mathcal{K}}(x)=0$ implies $x=0$.
\item $\mathcal{K}$ is absorbing, which (in this setting) is equivalent to $\mathcal{K}$ containing the origin in its interior.
\item $\mathcal{K}$ is balanced (which in this setting is equivalent to symmetry $\mathcal{K}=-\mathcal{K}$), so that we always have ${\sf m}_{\mathcal{K}}(x)={\sf m}_{\mathcal{K}}(-x)$.
\end{itemize}
In this paper, we are interested in the cases where $\mathcal{K}$ is indeed bounded, absorbing, and balanced, so we can define the 
$\mathcal{K}$ norm:
\begin{align}
\|x\|_{\mathcal{K}}:={\sf m}_{\mathcal{K}}(x).
\end{align} 
If $\mathcal{K}=\mathcal{C}^{\circ}$ is the \emph{polar} of some convex body $\mathcal{C}$:
\begin{align}
\mathcal{C}^{\circ}:=
\{z\in\mathbb{R}^n\colon \sup_{y\in\mathcal{C}}\left<y,z\right>\le 1\},
\end{align}
then we also have
\begin{align}
\|x\|_{\mathcal{C}^{\circ}}=\sup_{y\in\mathcal{C}}\left<y,x\right>.
\end{align}

\subsection{Information Measures}
Recall that the relative entropy between two probability distributions $P$ and $Q$ on the same measurable space (alphabet) is defined as 
\begin{align}
D(P\|Q):=\int \log\frac{dP}{dQ}dP.
\end{align}
Given a distribution $P_{XY}$ on $\mathcal{X}\times \mathcal{Y}$ and a conditional distribution $Q_{Y\vert X}$, define the conditional relative entropy
\begin{align}
D(P_{Y\vert X}\|Q_{Y\vert X}P_X):=D(P_{YX}\|Q_{Y\vert X}P_X)
\end{align}
where $Q_{Y\vert X}P_X$ is the joint distribution induced by $P_X$ and $Q_{Y\vert X}$.
The mutual information for a joint distribution $P_{XY}$ is defined as 
\begin{align}
I(X;Y):=D(P_{Y\vert X}\|P_Y\vert P_X)
\end{align}
and it is well-known that the channel capacity for a channel $P_{Y\vert X}$ equals $\sup_{P_X}I(X;Y)$ where $(X,Y)\sim P_{Y\vert X}P_X$.
Given $P_{XYZ}$, the conditional mutual information is given by 
\begin{align}
I(X;Y\vert Z):=\sum_xI(X;Y\vert Z=z)P_Z(z)
\end{align}
where $I(X;Y\vert Z=z)$ is the mutual information under the distribution $P_{XY\vert Z=z}$.
The entropy is defined by 
\begin{align}
H(X):=\sum_xP_X(x)\log \frac1{P_X(x)},
\end{align}
and conditional entropy is defined by 
\begin{align}
H(Y\vert X)&:=\sum_{x}H(Y\vert X=x)P_X(x)
\end{align}
where $H(Y\vert X=x)$ is the entropy of $Y$ under $P_{Y\vert X=x}$.

\section{Lower Bounding the Supremum via Mixed Volumes}\label{sec_glemma}
In this section we recall Pajor's bound mentioned in \eqref{e_intro1}.
Let $\mathcal{C}$ be a bounded, symmetric convex polyhedron in $\mathbb{R}^N$ containing the origin in its interior. 
Let $\mathcal{C}^{\circ}$ be its polar.
Note that each vertex $y$ of $\mathcal{C}$ corresponds to a facet,
a $(N-1)$-dimensional face of $\mathcal{C}^{\circ}$ \cite{grunbaum2013convex}.
With an abuse of notation we use $y$ to denote this facet;
denote by $h_y$ the distance from this facet to the origin, and $S_y$ the area of this facet.
Let $P_Y$ be the probability distribution on the vertices of $\mathcal{C}$ such that 
\begin{align}
P_Y(y):=\frac{h_yS_y}{N\vol(\mathcal{C}^{\circ})}.
\label{e_geometry}
\end{align}
The measure $P_Y$ is also called the cone volume measure.
Note that this is indeed a probability measure since $\frac1{N}h_yS_y$ is the volume of the convex hull spanned by the facet $y$ and the origin.
Note also that if $Y\sim P_Y$ then
\begin{align}
\mathbb{E}[Y]=0
\end{align}
because the symmetry of $\mathcal{C}$ implies that $Y$ and $-Y$ have the same distribution.

\begin{lem}\label{lem_general}(Pajor, \cite[p40, eq 2-7]{pajort})
Suppose that $\mathcal{C}$ and $P_Y$ are as defined above.
Let $\mathcal{A}\subseteq \mathbb{R}^N$ be compact,
and define 
\begin{align}
a:=\mathbb{E}[\sup_{z\in\mathcal{A}}\left<z,Y\right>
],
\end{align}
where $Y\sim P_Y$, and the inner product is the $N$-dimensional inner product.
We have:
\begin{itemize}
\item The $N$-dimensional Euclidean volume of $\mathcal{A}$ is bounded as
\begin{align}
\vol(\mathcal{A})
\le 
\vol(a\mathcal{C}^{\circ}).
\label{e_claim1}
\end{align}
\item Let $l>0$.
The $l$-packing number of $\mathcal{A}$ under $\|\|_{\mathcal{C}^{\circ}}$ satisfies
\begin{align}
{\sf P}_l(\mathcal{A})\le 
\left(1+\frac{2a}{l}\right)^N.
\label{e_claim2}
\end{align}
\end{itemize}
\end{lem}
Since the original reference \cite{pajort} is in French, we provide here the short proof for reader's convenience:
\begin{proof}
Since $\mathbb{E}[\sup_{z\in\mathcal{A}}\left<z,Y\right>
]$ is unchanged while the volume and the packing number do not decrease when $\mathcal{A}$ is replaced by its convex hull which is also compact, it suffices to consider convex compact $\mathcal{A}$.
Moreover it suffices to consider such $\mathcal{A}$ which also has nonempty interior (so that $\mathcal{A}$ is a convex body), since otherwise we can apply an approximation argument with $\mathcal{A}$ replaced by its Minkowski sum with a ball of vanishing radius.

Now by definition of the mixed volume we have
\begin{align}
N \cdot V(\mathcal{A},\mathcal{C}^{\circ},\dots,\mathcal{C}^{\circ})
&=V(\mathcal{A},\mathcal{C}^{\circ},\dots,\mathcal{C}^{\circ})
+V(\mathcal{C}^{\circ},\mathcal{A},\mathcal{C}^{\circ},\dots,\mathcal{C}^{\circ})
+
V(\mathcal{C}^{\circ},\dots,\mathcal{C}^{\circ},\mathcal{A})
\\
&=\sum_{y\in \,\textrm{facets of }\mathcal{C}^{\circ}}
\sup_{z\in\mathcal{A}}\left<n_y,z\right>
\cdot
S_y
\label{e23}
\\
&=
\sum_{y\in \,\textrm{facets of }\mathcal{C}^{\circ}}
\sup_{z\in\mathcal{A}}\left<\frac{1}{h_y}\,n_y,z\right>
\cdot
S_yh_y
\\
&=N\vol(\mathcal{C}^{\circ})\cdot \mathbb{E}[\sup_{z\in\mathcal{A}}\left<z,Y\right>
]
\end{align}
where we recall that $n_y$ denotes the outward normal of the facet $y$,
and \eqref{e23} follows by taking the limit in \eqref{e7}.
However,
by the Alexandrov-Fenchel Inequality,
\begin{align}
\vol^{\frac{1}{N}}(\mathcal{A})
\cdot
\vol^{\frac{N-1}{N}}(\mathcal{C}^{\circ})
&=V^{\frac{1}{N}}(\mathcal{A},\dots,\mathcal{A})
\cdot
V^{\frac{N-1}{N}}(\mathcal{C}^{\circ},\dots,\mathcal{C}^{\circ})
\\
&\le 
V(\mathcal{A},\mathcal{C}^{\circ},\dots,
\mathcal{C}^{\circ})
\\
&=\vol(\mathcal{C}^{\circ})
\mathbb{E}[\sup_{z\in\mathcal{A}}\left<z,Y\right>
]
\label{e_sv}
\\
&= a\vol(\mathcal{C}^{\circ})
\end{align}
Rearranging, we obtain
\begin{align}
\vol^{\frac{1}{N}}(\mathcal{A})
\le 
a \vol^{\frac{1}{N}}(\mathcal{C}^{\circ})
\\
=\vol^{\frac{1}{N}}(a\mathcal{C}^{\circ})
\end{align}
which proves \eqref{e_claim1}.

To show \eqref{e_claim2}, let 
\begin{align}
\mathcal{B}:=\frac{l}{2}\,\mathcal{C}^{\circ}
+\mathcal{A}.
\end{align}
Note that for any distinct elements $x_1$and $x_2$ in an $l$-packing of $\mathcal{A}$, 
we have $x_1-x_2\notin l\mathcal{C}^{\circ}$, 
and therefore $x_1+\frac{l}{2}\mathcal{C}^{\circ}$ and $x_2+\frac{l}{2}\mathcal{C}^{\circ}$ do not intersect. (This is the only part we used the symmetry of $\mathcal{C}$.) Thus
\begin{align}
\vol(\mathcal{B})
\ge {\sf P}_l(\mathcal{A})
\left(\frac{l}{2}\right)^N\vol(\mathcal{C}^{\circ}).
\label{e_union}
\end{align}
Moreover,
\begin{align}
\mathbb{E}[\sup_{z\in\mathcal{B}}\left<z,Y\right>
]
&=
\mathbb{E}\left[\frac{l}{2}\sup_{z\in\mathcal{C}^{\circ}}\left<z,Y\right>
+
\sup_{z\in\mathcal{A}}\left<z,Y\right>
\right]
\\
&=\frac{l}{2}+a.
\end{align}
Thus by applying \eqref{e_claim1} to the set $\mathcal{B}$, we have 
\begin{align}
\vol(\mathcal{B})\le \left(\frac{l}{2}+a\right)^N
\vol\left(\mathcal{C}^{\circ}\right)
\end{align}
which, together with \eqref{e_union}, implies \eqref{e_claim2}.
\end{proof}
For comparisons of 
Lemma~\ref{lem_general} and other minoration inequalities in the settings of Gaussian and Rademacher processes, the reader is referred to \cite{mendelson2019generalized} or an earlier version of this paper \cite[Section~3.2-3.3]{liu_minoration}.
In particular, by choosing $\mathcal{C}$
to be a ball one can obtain from Lemma~\ref{lem_general} a dimension dependent ``weak Sudakov minoration'' for the Gaussian process \cite{mendelson2019generalized}.
Further using a dimension reduction (Johnson-Lindenstrauss) argument,
one can obtain the standard dimension-free Sudakov minoration \cite{mendelson2019generalized}.
The combination of the weak Sudakov and the standard Sudakov is called 
``improved Sudakov minoration'' in \cite{mendelson2019generalized}.
We show in \cite[Section~3.2]{liu_minoration} that the improved Sudakov minoration is in fact equivalent to the standard Sudakov minoration upon optimizing the scale parameter.

\section{From Cover's Problem to Soft-Minoration}
\label{sec_application}
\subsection{Relay Channel and Cover's Problem}
\label{sec_cover}
Relay channel is a three-terminal communication model where a relay is present to help the transmission of a message from a transmitter to a receiver.
The general goal is to understand the maximum rate of transmitting a message to the receiver using the given relay such that the error probability is asymptotically vanishing.
In \cite{cover} (texts also reproduced in \cite{WuBarnesOzgur}), Thomas Cover considered a seemingly simply relay channel as depicted in Figure~\ref{f_m1} (later also referred to as the primitive relay channel \cite{kim_techniques}).
The mathematical formulation is as follows:

{\bf Model~1:} (Figure~\ref{f_m1}) Let $R,R_0\in(0,\infty)$
and fix two channels (i.e., conditional distributions) $P_{Y\vert X}$ and $P_{Z\vert X}$ from $\mathcal{X}$ to $\mathcal{Y}$ and $\mathcal{Z}$, respectively, where $\abs{\mathcal{X}},\abs{\mathcal{Y}},\abs{\mathcal{Z}}<\infty$.
For each nonnegative integer $n$:
\begin{itemize}
\item The message $W$ is a random variable  equiprobable on $\{1,\dots,\lfloor \exp(nR)\rfloor\}$;
\item The encoder is a map $\Phi\colon \{1,\dots,\lfloor \exp(nR)\rfloor\}\to \mathcal{X}^n$;
\item Given $X^n=x^n$, the channel outputs $Y^n$ and $Z^n$ are conditionally independent and follows the distribution $\prod_{i=1}^nP_{Y\vert X=x_i}\prod_{i=1}^nP_{Z\vert X=x_i}$;
\item The relay computes a relay message $V\in\{1,2,\dots,\lfloor\exp(nR_0)\rfloor\}$ from $Z^n$. 
With an abuse of notation we denote this map as $V:\mathcal{Z}^n\to \{1,2,\dots,\lfloor\exp(nR_0)\rfloor\}$.
\item The decoder computes $\hat{W}=\hat{W}(V,Y^n)\in\{1,2,\dots,\lfloor\exp(nR)\rfloor\}$ based on $V$ and $Y^n$. 
\item The error probability is $P_e^{(n)}:=\mathbb{P}[W\neq \hat{W}]$.
\end{itemize}
\begin{figure}[h]
  \centering
\begin{tikzpicture}
[node distance=0.6cm,minimum height=6mm,minimum width=6mm,arw/.style={->,>=stealth'}]
  \node[rectangle,draw,rounded corners] at (0,0) (E) {Encoder};
   \node[rectangle] [right =0.5cm of E](X) {$X^n$};
  \node[rectangle] (Y) [right =1cm of X] {$Y^n$};
  \node[rectangle] [above = 1cm of Y](Z){$Z^n$};
  \node[rectangle] (W) [left =0.4cm of E] {$W$};
  \node[rectangle,draw,rounded corners] (D) [right =2cm of Y] {Decoder};
  \node[rectangle,draw,rounded corners] (R) [right =0.4cm of Z] {Relay};
  \node[rectangle] (V) [right =0.4cm of R] {$V$};  
  \node[rectangle] (H) [right =0.4cm of D] {$\hat{W}$};

  \draw [arw] (W) to node[midway,above]{} (E);
  \draw [arw] (E) to node[midway,above]{} (X);
  \draw [arw] (X) to node[midway,above]{} (Y);  
  \draw [arw] (X) to node[midway,right]{} (Z);
  \draw [arw] (Z) to node[midway,right]{} (R);
  \draw [arw] (R) to node[midway,right]{} (V);  
  \draw [arw] (V) to node[midway,right]{} (D); 
  \draw [arw] (Y) to node[midway,right]{} (D);  
  \draw [arw] (D) to node[midway,right]{} (H);    
\end{tikzpicture}
\caption{Relay channel}
\label{f_m1}
\end{figure}
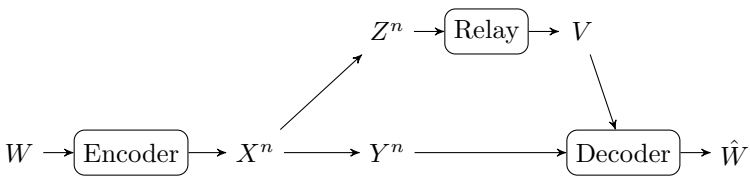

Given $R_0>0$ and $\epsilon\in(0,1)$, define
\begin{align}
C_{\epsilon}(R_0):=\sup\{R\colon \limsup_{n\to\infty}\inf P_e^{(n)}\le \epsilon\},
\end{align}
where the infimum is over all encoder $\Phi$, relay coder $V$, and decoder $\hat{W}$ in Model~1 with parameters $R_0$ and $R$.
The capacity of the relay channel is defined as
\begin{align}
C(R_0):=\lim_{\epsilon\downarrow 0}C_{\epsilon}(R_0)
\end{align}
where the limit exists due to monotone convergence.
In other words, $C(R_0)$ is the maximum rate $R$ at which the transmitter can transmit a message $W$ to the receiver, when the relay is allowed to send a relay message $V$ to the receiver at the rate $R_0$.

As with many problems in multi-terminal information theory, the relay channel is unsolved (in the sense of a single-letter characterization of $C(R_0)$ in terms of general $R_0$, $P_{Z\vert X}$, and $P_{Y\vert X}$), despite the simplicity of its formulation.
A question about a particular characteristic of $C(R_0)$, raised by Thomas Cover \cite{cover}, is the following:
\begin{quote}
\it
If $Y$ and $Z$ are conditionally identically distributed given $X$ (i.e., $P_{Y\vert X}=P_{Z\vert X}$), what is the minimum $R_0$ for which $C(R_0)=C(\infty)$? 
\end{quote}
To gain intuitions and to motivate the above question, Cover noted in \cite{cover} a few preliminary properties of $C(\cdot)$:
\begin{enumerate}
\item $C(0)=\max_{P_X}I(X;Y)$, 
which is the Shannon capacity \cite{shannon1948mathematical} of the point-to-point channel $P_{Y\vert X}$ from $\mathcal{X}$ to $\mathcal{Y}$;
\item $C(\infty)=\max_{P_X}I(X;Y,Z)$ where $P_{XYZ}=P_XP_{Y\vert X}P_{Z\vert X}$, since
when $R_0=\infty$ the relay can losslessly send $Z^n$ to the intended receiver and the problem reduces to a point-to-point channel from $\mathcal{X}$ to $\mathcal{Y}\times \mathcal{Z}$;
\item $C(\cdot)$ is a nondecreasing function.
\end{enumerate}
It is unclear why Cover originally imposed the symmetry assumption $P_{Y\vert X}=P_{Z\vert X}$ since the problem is valid even without it.
We shall use it in justifying some regularity conditions; see \eqref{e_Delta} \eqref{e_full} (it boils down the fact that for a rank 2 matrix $A$ and another arbitrary matrix $B$, their product may be rank 1, but $AA^{\top}$ always remains rank 2).
However, some other condition may be imposed in place of symmetry to retain regularity and clean answers (see the end of Section~\ref{sec_general}).

Several achievability schemes for the primitive relay channel have been studied in \cite{kim_techniques}, yielding various lower bounds on $C(\cdot)$.
In particular, the compress-and-forward scheme gives 
\begin{align}
C(R_0)\ge \max_{P_X,P_{\hat{Z}\vert Z}}
\{I(X;Y,\hat{Z})\colon I(Z;\hat{Z}\vert Y)\le R_0\}.
\label{e_caf}
\end{align}
Roughly speaking, in compress-and-forward the relay sends $V$ at the rate $I(Z;\hat{Z}\vert Y)$ using the Wyner-Ziv scheme so that the receiver can simulate $\hat{Z}^n$, essentially creating a channel $P_{\hat{Z}Y\vert X}$ to the intended receiver \cite{kim_techniques}.
While \eqref{e_caf} may not be tight in general, it will be seen from our main result that it is tight for the critical $R_0$ in Cover's problem.

A major challenge in solving Cover's problem is to derive tight converse bounds (upper bounds on $C(\cdot)$).
Despite the many past and recent efforts using various high dimensional probability tools, Cover's problem was previously only solved for the special case of binary symmetric channel \cite{BarnesWuOzgur-BSC} or the Gaussian version \cite{WuBarnesOzgur}.
 
Before explaining our solution to Cover's problem, let us look at a few concrete examples of channels.

\begin{example}\label{exp1}
 Let $P_{Z\vert X}$ be arbitrary and let $Y$ always be a constant.
Then $R_{\rm crit}:=\inf\{R_0\colon C(R_0)=C(\infty)\}=\sup_{P_X}I(X;Z)$ is the capacity of the channel to $Z$.
\end{example}
It was shown in \cite{BarnesWuOzgur-BSC},
using rearrangement inequalities on the Hamming sphere, that for binary symmetric channels, 
we have $R_{\rm crit}=H(Z\vert Y)$ where $Y$ and $Z$ follow the capacity-achieving distribution.
Since $I(X;Z)$ can be strictly smaller than $H(Z)=H(Z\vert Y)$,
 Example~\ref{exp1} shows that $R_{\rm crit}=H(Z\vert Y)$ cannot be true in general.
The reader may argue that symmetry $P_{Z\vert X}=P_{Y\vert X}$ does not hold in Example~\ref{exp1}.
The following example does preserve symmetry:

\begin{example}\label{exp2}
Let $P_{Z\vert X=x}=P_{Y\vert X=x}$ be independent of $x$.
Then $R_{\rm crit}=0$ as the channel capacity is 0. 
\end{example}
Still $R_{\rm crit}<H(Z\vert Y)$.
The reader may argue that in Example~\ref{exp2}, a constant is a sufficient statistic of $Z$ for $X$,
so perhaps $R_{\rm crit}=H(\uZ\vert Y)$ where $\uZ$ is a certain ``most succinct sufficient statistic'' for $X$?
However, we can not find such a sufficient statistic for $X$ in Example~\ref{exp1} so that $R_{\rm crit}=H(\uZ\vert Y)$.
In fact, we can also construct symmetry examples in which $R_{\rm crit}<H(\uZ\vert Y)$ for any sufficient statistic $\uZ$ of $Z$ for $X$:

\begin{example}\label{exp3}
Consider $P_{Y\vert X}=P_{Z\vert X}$ where $\mathcal{X}=\{1,2,3\}$, and
\begin{align}
P_{Z\vert X=1}&=[\frac1{4}+\frac{\epsilon}{2},\,
\frac1{4}+\frac{\epsilon}{2},\,
\frac1{2}-\epsilon];
\\
P_{Z\vert X=2}&=[\frac1{4}-\frac{\epsilon}{2},\,
\frac1{4}-\frac{\epsilon}{2},\,
\frac1{2}+\epsilon];
\\
P_{Z\vert X=3}&=[\frac1{4}-\delta,\,
\frac1{4}+\delta,\,
\frac1{2}];
\end{align}
Then there exists $c>0$ small enough such that for any $\epsilon\in(0,c)$ and $\delta\in(0,\epsilon^2)$, we have that $P_X=[\frac1{2},\frac1{2},0]$ maximizes $I(X;YZ)$.
We cannot combine symbols in $\mathcal{Z}$ to form a ``more succinct'' sufficient statistic for $X$.
Yet using our main result (Theorem~\ref{thm1}), we will find $R_{\rm crit}={\rm H}(\frac1{2}+2\epsilon^2)$, where ${\rm H}(\cdot)$ denotes the binary entropy function. In contrast, $H(Z\vert Y)=\frac1{2}\log2+{\rm H}(\frac1{2}+2\epsilon^2)$.
\end{example}
The claims in Example~\ref{exp3} will be shown in Appendix~\ref{app1}.

\subsection{Solution for General Channels}\label{sec_general}
Let $\uY$ and $\uZ$ be functions of $Y$ and $Z$ respectively, whose definitions are postponed to Definition~\ref{defn1} since they require the notion of capacity-achieving output distribution.
Our general solution to Cover's problem is that $R_{\rm crit}:=\inf\{R_0\colon C(R_0)=C(\infty)\}=H(\uZ\vert \uY)$ where $(Y,Z)$ follows the capacity-achieving output distribution.
The upper bound part is the following statement:
\begin{prop}\label{prop_achieve}
$C(H(\uZ\vert \uY))=C(\infty)$.
\end{prop}
The proof of Proposition~\ref{prop_achieve} is given in Appendix~\ref{app_prop_achieve}, which is relatively immediate from the compress-and-forward bound \eqref{e_caf}.
The converse part, on the other hand, is nontrivial and is our main result:
\begin{thm}\label{thm1}
Given $P_{Y\vert X}=P_{Z\vert X}$, there exists $c>0$ such that 
for any $R_0\in[H(\uZ\vert \uY)-c^{-1},H(\uZ\vert \uY)]$, 
we have
\begin{align}
H(\uZ\vert \uY)-R_0
\le
c\lambda^{\frac1{10}}\log^{\frac{6}{5}}\frac1{\lambda}
\end{align}
where $\lambda:=C(\infty)-C(R_0)$.
In particular, $R_{\rm crit}=H(\uZ\vert \uY)$.
\end{thm}
\begin{proof}
Suppose that for parameters $n$, $R\le \log\abs{\mathcal{X}}$, $R_0$ and $\epsilon$, Model~1 can achieve error $P_e^{(n)}\le \epsilon$.
In the next a few sections we shall introduce functions $\mu$ \eqref{e_mu}, $\Delta$ \eqref{e_Delta} depending only on $P_{Z\vert X}$, and functions $g_n$, $f_n$, $\bar{f}_n$, and $h_n$, depending on $P_{Z\vert X}$ and $n$, all on $(0,\infty)$, 
such that 
\begin{align}
H(\uZ\vert \uY)-R_0&\le
g_n(\tau)+\tau+O(\sqrt{\lambda'})
&\textrm{(Lemma~\ref{lem6})};&
\\
g_n(t)&\le 2\max\{t,f_n(t)\}
&\textrm{(Lemma~\ref{lem7})};&
\\
f_n(t)&\le \bar{f}_n(t^{4/5})+t^{\frac1{5}}\log\abs{\mathcal{Z}}+O(\frac1{n}\log n)
&\textrm{(Lemma~\ref{lem8})};&
\\
\bar{f}_n(t)&\le h_n(t)
&\textrm{(Lemma~\ref{lem11})};&
\\ 
\limsup_{n\to\infty}h_n(t)&=
O(t^{\frac1{4}}\log\frac1{t})
&\textrm{(Lemma~\ref{lem_l1})}.&
\end{align}
where $\lambda':=C(\infty)-R+\mu(\epsilon)\ge 0$,
$\tau:=\mu(\Delta(\lambda'))+\lambda'+\tfrac{\log n}{n}$,
and $t\in (0,1/2)$ is arbitrary; all $O(\cdot)$ hide multiplicative factors depending only on $P_{Z\vert X}$, and their arguments are assumed to be close to 0.
From \eqref{e_mu} and \eqref{e_Delta} we have $\mu(t)=O(t\log\frac1{t})$ and $\Delta(t)=O(\sqrt{t})$.
Using these and taking $n\to\infty$ we obtain
\begin{align}
H(\uZ\vert \uY)-R_0
=
O\left((\lambda')^{\frac1{10}}\log^{\frac{6}{5}}\frac1{\lambda'}\right).
\end{align}
Taking $R\uparrow C_{\epsilon}(R_0)$ gives $H(\uZ\vert \uY)-R_0
=
O\left((\lambda'')^{\frac1{10}}\log^{\frac{6}{5}}\frac1{\lambda''}\right)$ where $\lambda'':=C(\infty)-C_{\epsilon}(R_0)+\mu(\epsilon)$.
Further taking $\epsilon\to0$ gives the desired estimate.
\end{proof}

The symmetry assumption $P_{Y\vert X}=P_{Z\vert X}$ is used in some steps in the proof of Theorem~\ref{thm1} (see \eqref{e_Delta} and \eqref{e_full}). 
However, some other regularity conditions may be imposed instead of symmetry for $R_{\rm crit}=H(\uZ\vert \uY)$ to hold, which we now discuss.

\begin{thm}\label{thm_general}
Let $\mathcal{X}$, $\mathcal{Y}$, $\mathcal{Z}$ bet finite sets, and consider arbitrary $P_{Y\vert X}$ and $P_{Z\vert X}$ (not necessarily $P_{Y\vert X}=P_{Z\vert X}$). 
Suppose that 
\begin{itemize}
\item The $z$-equivalence classes are singletons. 
That is, $z\mapsto \underline{z}$ is injective.
\item $P_{Y\vert X=x}(y)>0$ for any $x\in\mathcal{X}$ and $y\in\mathcal{Y}$. 
\end{itemize}
Then $R_{\rm crit}=H(Z\vert Y)$, where $(Y,Z)$ follows $Q_{YZ}$.
\end{thm}
The proof is given in Appendix~\ref{app_thm_general}.
Theorem~\ref{thm_general} implies the following:

\begin{thm}\label{thm3}
Suppose that $2\le \abs{\mathcal{X}}\le \abs{\mathcal{Y}}<\infty$ and $\abs{\mathcal{Z}}<\infty$.
Then we have $R_{\rm crit}=H(Z\vert Y)$
where $(Y,Z)$ follows $Q_{YZ}$, 
for almost all $(P_{Y\vert X}, P_{Z\vert X})\in (\Delta^{\abs{\mathcal{Y}}-1})^{\abs{\mathcal{X}}}\times (\Delta^{\abs{\mathcal{Z}}-1})^{\abs{\mathcal{X}}}$ (with respect to the measure induced by the metric, where $\Delta^d$ denotes the $d$-dimensional probability simplex).
\end{thm}
The proof is given in Appendix~\ref{app_thm3}.
\begin{rem}
After an initial version of this paper appeared on arXiv, El Gamal, Gohari and Nair also posted their concurrent work \cite{gamal_cut},
which employed the traditional auxiliary random variable approach and showed that when $\{P_{Y\vert X}(y\vert x)\}_{x,y}$ has full row-rank (meanwhile, the symmetry assumption $P_{Y\vert X}=P_{Z\vert X}$ can be dropped), compress-forward is optimal for achieving $R_{\rm crit}$.
\end{rem}

\subsection{The Capacity-Achieving Output Distribution}
\label{sec_caod}
For a general discrete memoryless channel $P_{Y\vert X}$,
the channel capacity is
\begin{align}
\sup_{P_X}I(X;Y)
&=\sup_{P_X}\inf_{Q_Y}D(P_{Y\vert X}\| Q_Y\vert P_X)
\label{e_mid}
\\
&=\inf_{Q_Y}\sup_{P_X}D(P_{Y\vert X}\|Q_Y\vert P_X)
\label{e_saddle}
\\
&=\inf_{Q_Y}\sup_xD(P_{Y\vert X=x}\|Q_Y).
\label{e_linear}
\end{align}
The steps are explained as follows: \eqref{e_mid} follows from $I(X;Y)=\inf_{Q_Y}D(P_{Y\vert X}\|Q_Y\vert P_X)
$. 
\eqref{e_saddle} and \eqref{e_linear} follow since $D(P_{Y\vert X}\|Q_Y\vert P_X)$ is linear in $P_X$ and convex in $Q_Y$.
These equivalent formulations of the capacity are known as the saddle point characterizations;
see \cite{minimax}.

The $Q_Y$ achieving the infimum in \eqref{e_saddle} is called the \emph{capacity-achieving output distribution}, which is known to be unique due to strong convexity. 
On the other hand there may be multiple $P_X$ achieving the supremum in \eqref{e_mid} unless additional assumptions such as full-rankness is imposed.

Now return to the relay channel problem.
Let us recall a classical argument to show that if the communication rate $R$ is close to the maximum capacity $C(\infty)$ and the error probability $\epsilon$ is small enough, then the output distribution $P_{Y^nZ^n}$ must be close to the capacity-achieving distribution.
Indeed, suppose that for parameters $n$, $R\le \log \abs{\mathcal{X}}$, $R_0$ and $\epsilon$, Model~1 can achieve error $P_e^{(n)}\le \epsilon$; we have
\begin{align}
nR
&\le I(X^n;V,Y^n)+n\mu(\epsilon)
\label{e_fano}
\\
&\le 
I(X^n;Z^n,Y^n)+n\mu(\epsilon)
\label{e_dpi}
\\
&= D(P_{Y^nZ^n\vert X^n}\|Q_{Y^nZ^n}\vert P_{X^n})
-D(P_{Y^nZ^n}\|Q_{Y^nZ^n})
+n\mu(\epsilon)
\\
&\le
nC_{\epsilon}(\infty)-D(P_{Y^nZ^n}\|Q_{Y^nZ^n})+n\mu(\epsilon)
\label{e_golden}
\end{align}
where 
\begin{itemize}
\item In \eqref{e_fano} we defined the function $\mu$ by
\begin{align}
\mu(\epsilon):=\epsilon\log\abs{\mathcal{X}}+{\rm H}(\epsilon)
\label{e_mu}
\end{align}
where ${\rm H}(\cdot)$ is the binary entropy function.
\eqref{e_fano} follows from Fano's inequality and the assumption $R\le \log\abs{\mathcal{X}}$.
\item \eqref{e_dpi} follows from the data processing inequality.
\item In \eqref{e_golden}, $Q_{Y^nZ^n}=Q_{YZ}^{\otimes n}$ and $Q_{YZ}$ is the capacity-achieving output distribution for the channel $P_{YZ\vert X}$.
\end{itemize}
Rearranging, 
and noticing that for point to point channel coding the capacity does not depend on the error probability (due to the strong converse \cite{wolfowitz1968}), 
we find $C_{\epsilon}(\infty)=C(\infty)=\sup_{P_X}I(X;YZ)$,
and obtain
\begin{align}
D(P_{Y^nZ^n}\|Q_{Y^nZ^n})
\le n[C(\infty)-R+\mu(\epsilon)]
\label{e87}
\end{align}
as desired.

Throughout this paper we will make the following assumption:
\begin{assump}\label{assump1}
For each $y\in\mathcal{Y}$, there exists $x\in\mathcal{X}$ such that $P_{Y\vert X=x}(y)>0$.
The same property holds for each $z\in\mathcal{Z}$.
\end{assump}
This assumption is without any loss of generality, since if $y$ is such that $P_{Y\vert X=x}(y)=0$ for all $x$, we can simply redefine $\mathcal{Y}$ by removing $y$.
Now from the saddle point condition we must have
\begin{align}
\{(y,z)\colon \max_xP_{Y\vert X=x}(y)P_{Z\vert X=x}(z)>0\}
\subseteq
\{(y,z)\colon Q_{YZ}(y,z)>0\}.
\label{e111}
\end{align} 
Indeed, otherwise we obtain that for some $x\in\mathcal{X}$, $\infty=D(P_{YZ\vert X=x}\|Q_{YZ})\le C(\infty)<\infty$, a contradition.
Now \eqref{e111} and Assumption~\ref{assump1} implies 
\begin{align}
Q_Y(y)>0,\,\forall y\in\mathcal{Y};
\quad 
Q_Z(z)>0,\,\forall z\in\mathcal{Z}.
\label{e_cond}
\end{align}

Next we define the function ``underline'' that we promised in the statement in the main result (Theorem~\ref{thm1}).
\begin{defn}\label{defn1}
Given $P_{Y\vert X}$ and $P_{Z\vert X}$ from Model~1, let $Q_{YZ}$ be the capacity-achieving output distribution associated with the channel $P_{Y\vert X}P_{Z\vert X}$.
We say two elements $z_1,z_2\in\mathcal{Z}$ are \emph{equivalent} if $Q_{Y\vert Z=z_1}=Q_{Y\vert Z=z_2}$ (where the conditional distributions are well-defined in view of \eqref{e_cond}).
Let $\uz$ denote the equivalent class of an arbitrary $z\in\mathcal{Z}$,
and let $\mathcal{\uZ}$ be the set of such equivalent classes.
Define $\uy$ and $\mathcal{\uY}$ similarly.
\end{defn}

\begin{prop}\label{prop7}
Given arbitrary $P_{Y\vert X}$ and $P_{Z\vert X}$, let $Q_{YZ}$ be the capacity-achieving output distribution.
Let $\uY$ and $\uZ$ be functions of $Y$ and $Z$ as defined by Definition~\ref{defn1}.
Then
\begin{enumerate}
\item $\uy_1=\uy_2$ if and only if $Q_{Z\vert \uY=\uy_1}= Q_{Z\vert \uY=\uy_2}$.
\item $\uy_1=\uy_2$ if and only if $Q_{\uZ\vert \uY=\uy_1}= Q_{\uZ\vert \uY=\uy_2}$.
\end{enumerate}
\end{prop}
\begin{proof}
\begin{enumerate}
\item We only need to show the ``if'' part.
Suppose that $\uy_1$ and $\uy_2$ are equivalent classes represented by $y_1$ and $y_2$, and $\uy_1\neq\uy_2$.
By definition we have that
\begin{align}
Q_{Z\vert Y=y}
=Q_{Z\vert Y=y_1}
\neq
Q_{Z\vert Y=y_2}
=Q_{Z\vert Y=y'}
\end{align}
for any $y\in \uy_1$ and $y'\in\uy_2$.
Then
\begin{align}
Q_{Z\vert \uY=\uy_1}&=\sum_{y\in\uy_1}Q_{Z\vert Y=y}
Q_{Y\vert \uY=\uy_1}(y)
\\
&=\sum_{y\in\uy_1}
Q_{Z\vert Y=y_1}
Q_{Y\vert \uY=\uy_1}(y)
\\
&=Q_{Z\vert Y=y_1}.
\label{e114}
\end{align}
Similarly, $Q_{Z\vert \uY=\uy_2}=Q_{Z\vert Y=y_2}$, so $Q_{Z\vert \uY=\uy_1}\neq Q_{Z\vert \uY=\uy_2}$.
\item Again it suffices to show the ``if'' part.
In view of Part~(1), it suffices to show the following statement:
\begin{align}
Q_{\uZ\vert \uY=\uy_1}=Q_{\uZ\vert \uY=\uy_2}
\Longrightarrow 
Q_{Z\vert \uY=\uy_1}=Q_{Z\vert \uY=\uy_2}
\label{e115}
\end{align}
for arbitrary $y_1$ and $y_2$.
To show \eqref{e115},
note that by the same argument in \eqref{e114},
we have 
\begin{align}
Q_{Y\vert Z}(y_1\vert z)=Q_{Y\vert \uZ}(y_1\vert \uz),
\quad \forall y_1\in\mathcal{Y},z\in\mathcal{Z},
\end{align}
and since $\uY$ is a function of $Y$ we have
\begin{align}
Q_{\uY\vert Z}(\uy_1\vert z)=Q_{\uY\vert \uZ}(\uy_1\vert \uz),
\quad \forall y_1\in\mathcal{Y},z\in\mathcal{Z}.
\label{e117}
\end{align}
Upon rearrangements, \eqref{e117} is equivalent to
\begin{align}
Q_{Z\vert \uY=\uy_1}(z)=
Q_{\uZ\vert \uY=\uy_1}(\uz)\cdot\frac{Q_Z(z)}{Q_{\uZ}(\uz)}
\end{align}
from which the validity of \eqref{e115} is immediate.
\end{enumerate}
\end{proof}

As mentioned, the fact that $z_1,z_2$ being in the same equivalent class does not necessarily imply that $\uZ$ is a sufficient statistic of $Z$ for $X$,
however we will show that we can eliminate a portion of $\mathcal{X}$ so that the implication holds.
Let us introduce the following notations:
\begin{align}
\mathcal{X}_{\rm bad}&:=
\left\{x\in\mathcal{X}\colon\exists z\neq z' \textrm{such that $\uz=\uz'$ but 
$\frac{P_{Z\vert X=x}(z)}{Q_Z(z)}\neq
\frac{P_{Z\vert X=x}(z')}{Q_Z(z')}$}\right\}
\label{e_bad}
\\
\mathcal{X}_{\rm good}&:=\mathcal{X}_{\rm bad}^c.
\label{e_good}
\end{align}
Next, we will show that the probability of $\mathcal{X}_{\rm bad}$ must be small if the output distribution is close to capacity-achieving.

\begin{prop}\label{prop_delta}
Suppose that $P_{Y\vert X}=P_{Z\vert X}$. 
\begin{enumerate}
\item Define function $\Delta(\cdot)$ such that for each $\delta>0$,
\begin{align}
\Delta(\delta):=\sup_{S_X}\{S_X(\mathcal{X}_{\rm bad})\}
\label{e_Delta}
\end{align}
where the supremum is over all distribution $S_X$ on $\mathcal{X}$ such that 
\begin{align}
D(\sum_{x\in\mathcal{X}}P_{YZ\vert X=x}S_X(x)\|
Q_{YZ})\le \delta.
\end{align}
Then $\Delta(\delta)=O(\sqrt{\delta})$ as $\delta\to0$. 
\item The capacity of the channel $P_{\uY\uZ\vert X}$ equals the capacity of the original channel $P_{YZ\vert X}$, with the capacity-achieving output distribution $Q_{\uY\uZ}$\footnote{A priori, $Q_{\uY\uZ}$ is the distribution induced by $Q_{YZ}$, the capacity-achieving output distribution for $P_{YZ\vert X}$, and the functions $Y\mapsto\uY$ and $Z\mapsto \uZ$. This proposition shows that the notation also coincides with the capacity-achieving output distribution for the channel $P_{\uY\uZ\vert X}$.}.
\end{enumerate}
\end{prop}
\begin{proof}
\begin{enumerate}
\item Suppose that $\Delta(\delta)=\Delta_0$.
By compactness of the probability simplex and the lower semicontinuity of the relative entropy, we can find a sequence $S_X^{(1)}$, $S_X^{(2)}$ \dots converging to $S_X^{\star}$ such that
\begin{align}
S_X^{\star}(\mathcal{X}_{\rm bad})
&= 
\lim_{j\to\infty}S_X^{(j)}(\mathcal{X}_{\rm bad})
= \Delta_0;
\\
D(\sum_{x\in\mathcal{X}}P_{YZ\vert X=x}S_X^{\star}(x)\|
Q_{YZ})&\le 
\liminf_{j\to\infty}D(\sum_{x\in\mathcal{X}}P_{YZ\vert X=x}S_X^{(j)}(x)\|
Q_{YZ})\le \delta.
\end{align}
Then there exists some $x_0\in\mathcal{X}_{\rm bad}$ such that 
\begin{align}
S_X^{\star}(x_0)>\frac{\Delta_0}{\abs{\mathcal{X}}}.
\label{e133}
\end{align}
By definition we can then find $z\neq z'$ satisfying $\uz=\uz'$ but 
$\frac{P_{Z\vert X=x_0}(z)}{Q_Z(z)}\neq
\frac{P_{Z\vert X=x_0}(z')}{Q_Z(z')}$.
Next we shall use the symmetry $P_{Y\vert X}=P_{Z\vert X}$.
The fact that $\uz=\uz'$ means $Q_{Y\vert Z=z}=Q_{Y\vert Z=z'}$, hence
\begin{align}
\frac{Q_{YZ}(z,z)}{Q_Z(z)}
&=\frac{Q_{YZ}(z,z')}{Q_Z(z')};
\label{e_83}
\\
\frac{Q_{YZ}(z',z)}{Q_Z(z)}
&=\frac{Q_{YZ}(z',z')}{Q_Z(z')}.
\label{e_84}
\end{align}
Next, we note that $Q_{YZ}(y,z)=Q_{YZ}(z,y)$ for any $(y,z)\in\mathcal{Y}\times \mathcal{Z}$ since $P_{Y\vert X}=P_{Z\vert X}$ and $Q_{YZ}$ is induced by $P_{Y\vert X}P_{Z\vert X}$ and some $P_X$. 
Therefore the following is a symmetric matrix
\begin{align}
\left(
\begin{array}{cc}
Q_{YZ}(z,z)  &   Q_{YZ}(z,z')  \\
Q_{YZ}(z',z)  &   Q_{YZ}(z',z')   
\end{array}
\right).\label{e85}
\end{align}
Using \eqref{e_83} and \eqref{e_84} we see that \eqref{e85}
is a scalar multiple of the following rank-1 matrix 
\begin{align}
\left(
\begin{array}{cc}
Q_Y(z)Q_Z(z)  &   Q_Y(z)Q_Z(z')  \\
Q_Y(z')Q_Z(z)  &   Q_Y(z')Q_Z(z')   
\end{array}
\right)
\end{align}
and hence
\begin{align}
\frac{Q_{YZ}(z,z)}{Q_Y(z)Q_Z(z)}
+
\frac{Q_{YZ}(z',z')}{Q_Y(z')Q_Z(z')}
-2\cdot\frac{Q_{YZ}(z,z')}{Q_Y(z)Q_Z(z')}=0.
\label{e129}
\end{align}
On the other hand, using $P_{Y\vert X}=P_{Z\vert X}$ and by completing the square we see that
\begin{align}
&\frac{P_{YZ\vert X=x}(z,z)}{Q_Y(z)Q_Z(z)}
+
\frac{P_{YZ\vert X=x}(z',z')}{Q_Y(z')Q_Z(z')}
-2\cdot\frac{P_{YZ\vert X=x}(z,z')}{Q_Y(z)Q_Z(z')}
\nonumber\\
&=
\left(
\frac{P_{Z\vert X=x}(z)}{Q_Z(z)}-\frac{P_{Z\vert X=x}(z')}{Q_Y(z')}
\right)^2\ge0
\end{align}
for all $x\in\mathcal{X}$,
with strict inequality when $x=x_0$, and the gap to 0 is a constant depending only on $P_{Y\vert X}$ (not on $\Delta_0$).
Integrating with respect to $S_X^{\star}$ and using \eqref{e133}, we have 
\begin{align}
\frac{S_{YZ}(z,z)}{Q_Y(z)Q_Z(z)}
+
\frac{S_{YZ}(z',z')}{Q_Y(z')Q_Z(z')}
-2\cdot\frac{S_{YZ}(z,z')}{Q_Y(z)Q_Z(z')}\ge c\Delta_0
\label{e131}
\end{align}
where $S_{YZ}:=\sum_{x\in\mathcal{X}}P_{YZ\vert X=x}S_X^{\star}(x)$ and $c$ is a constant depending only on $P_{Y\vert X}$ (not on $\Delta_0$).
Comparing \eqref{e129} and \eqref{e131} we obtain the lower bound the total variation distance:
\begin{align}
\|S_{YZ}-Q_{YZ}\|_{TV}\ge c'\Delta_0
\end{align}
and in turn,
\begin{align}
\delta:=D(S_{YZ}\|Q_{YZ})\ge c''\Delta_0^2
\end{align}
where $c',c''\in (0,\infty)$ are constants depending only on $P_{Y\vert X}$.
This establishes that $\Delta(\delta)=\Delta_0=O(\sqrt{\delta})$.
 
\item
Suppose that $P_X$ is an arbitrary capacity achieving input distribution for the channel $P_{YZ\vert X}$, so that 
\begin{align}
\sum_{x\in\mathcal{X}}P_{YZ\vert X=x}P_X(x)=Q_{YZ}. \label{e_px}
\end{align}
Then part (1) implies that $P_X(\mathcal{X}_{\rm bad})=0$.
Consider arbitrary $x$ on the support of $P_X$, so that $x\in\mathcal{X}_{\rm good}$, and arbitrary $(y,z)$ and $(y',z')$ satisfying $P_{YZ\vert X=x}(y,z)>0$, $P_{YZ\vert X=x}(y',z')>0$, and $\uz=\uz'$, $\uy=\uy'$.
Then \eqref{e111} implies $Q_{YZ}(y,z)$ and $Q_{YZ}(y',z')>0$, and $x\in\mathcal{X}_{\rm good}$ implies
\begin{align}
\frac{P_{Z\vert X=x}(z)}{Q_Z(z)}=
\frac{P_{Z\vert X=x}(z')}{Q_Z(z')},
\quad
\frac{P_{Y\vert X=x}(y)}{Q_Y(y)}=
\frac{P_{Y\vert X=x}(y')}{Q_Y(y')}.
\end{align}
Hence 
\begin{align}
\frac{P_{YZ\vert X=x}(y,z)}{Q_Y(y)Q_Z(z)}
=\frac{P_{YZ\vert X=x}(y',z')}{Q_Y(y')Q_Z(z')}.
\end{align}
However,
\begin{align}
Q_Y(y)Q_Z(z)&=Q_Y(y)\frac{Q_{YZ}(y,z)}{Q_{Y\vert Z}(y\vert z)}
\\
&=Q_Y(y)\frac{Q_{YZ}(y,z)}{Q_{Y\vert Z}(y\vert z')}
\\
&=Q_Y(y)\frac{Q_{YZ}(y,z)Q_Z(z')}{Q_{YZ}(y,z')}
\\
&=\frac{Q_{YZ}(y,z)Q_Z(z')}{Q_{Z\vert Y}(z'\vert y)}
\\
&=\frac{Q_{YZ}(y,z)Q_Z(z')}{Q_{Z\vert Y}(z'\vert y')}
\\
&=\frac{Q_{YZ}(y,z)Q_Z(z')Q_Y(y')}{Q_{YZ}(y',z')},
\end{align}
therefore,
\begin{align}
\frac{P_{YZ\vert X=x}(y,z)}{Q_{YZ}(y,z)}
=\frac{P_{YZ\vert X=x}(y',z')}{Q_{YZ}(y',z')}.
\end{align}
Using the elementary identity $\frac{A}{B}=\frac{C}{D}\Rightarrow\frac{A}{B}=\frac{A+C}{B+D}$, the above implies that 
\begin{align}
\frac{P_{YZ\vert X=x}(y,z)}{Q_{YZ}(y,z)}
=\frac{P_{\uY\uZ\vert X=x}(\uy,\uz)}{Q_{\uY\uZ}(\uy,\uz)}.
\label{e102}
\end{align}
Since $\sum_{x\in\mathcal{X}}P_{\uY\uZ\vert X=x}P_X(x)=Q_{\uY\uZ}$ by \eqref{e_px},
taking the expectation of \eqref{e102} with respect to $P_XP_{YZ\vert X}$ we obtain
$I(X;\uY\uZ)=I(X;YZ)$.
Therefore the capacity of the channel $P_{\uY\uZ\vert X}$ is no less than the capacity of the original channel $P_{YZ\vert X}$,
and so they must be equal.
Moreover $P_X$ is a capacity-achieving input distribution while $Q_{\uY\uZ}$ is the capacity-achieving output distribution.
\end{enumerate}
\end{proof}
\begin{rem}
Proposition~\ref{prop_delta} may fail when the assumption $P_{Y\vert X}=P_{Z\vert X}$ is removed.
Indeed, this is implied by the following example where there exists $x_0\in\mathcal{X}_{\rm bad}$ in the support of a capacity achieving input distribution: Consider the case where $\mathcal{X}=\{x_0,x_1,x_2,x_3\}$, $Y$ and $Z$ are binary, and 
\begin{align}
P_{Y\vert X=x_0}=P_{Y\vert X=x_1}=[1,0];
\\
P_{Y\vert X=x_2}=P_{Y\vert X=x_3}=[0,1];
\\
P_{Z\vert X=x_0}=P_{Z\vert X=x_2}=[1,0];
\\
P_{Z\vert X=x_1}=P_{Z\vert X=x_3}=[0,1].
\end{align}
We can see that the capacity achieving output distribution is $Q_{YZ}$ equiprobable on $\mathcal{Y}\times \mathcal{Z}$, achieved when $X$ is equiprobable on $\mathcal{X}$. Thus the two elements in $\mathcal{Z}$ are in the same equivalent class, but $x_0\in\mathcal{X}_{\rm bad}$.
\end{rem}

We can reduce the channel $P_{YZ\vert X}$ to $P_{\uY\uZ\vert X}$ if the input symbols are restricted to $\mathcal{X}_{\rm good}$.
In reality such a restriction is not in place; however the probability of $\mathcal{X}_{\rm bad}$ must be small if the scheme is nearly capacity-achieving.
Therefore with a small bit of oracle message taking care of the symbols in $\mathcal{X}_{\rm bad}$, a good scheme for Model~1 can be converted to a good scheme for the following modified model which is essentially for the channel $P_{\uY\uZ\vert X}$:

{\bf Model~2: } (Figure~\ref{f_m2})
Suppose that $P_{Y\vert X}$ and $P_{Z\vert X}$ are given.
The model is similar to Model~1, but the channel outputs are the equivalent classes $\uY^n$ and $\uZ^n$.
Moreover, the relay and the decoder receive an oracle information $E=(E_1,\dots,E_n)$, which is a function of $x^n$: 
\begin{align}
E_i&=x_i,\quad\textrm{if $x_i\in\mathcal{X}_{\rm bad}$};
\\
E_i&=*, \quad\textrm{if $x_i\in\mathcal{X}_{\rm good}$},
\end{align}
where $*$ is a dummy symbol.
(Alternatively, $E$ may be defined as the equivalent class of $x^n$ for which coordinates in $\mathcal{X}_{\rm bad}$ agree.)
The outputs of the relay and the decoder are given by the functions $V=V(\uZ^n,E)\in\{1,2,\dots,\lfloor\exp(nR_0)\rfloor\}$ and $\hat{W}=\hat{W}(V,\uY^n,E)\in \{1,2,\dots,\lfloor\exp(nR)\rfloor\}$.

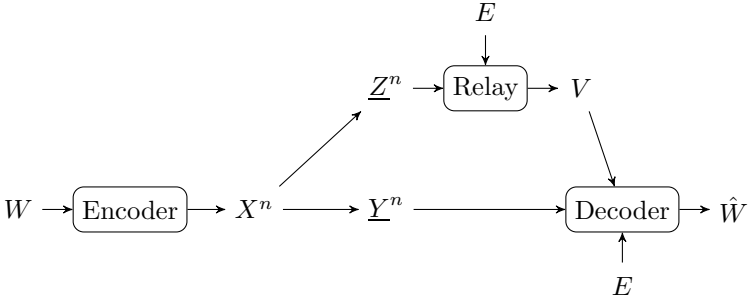
\begin{figure}[h]
  \centering
\begin{tikzpicture}
[node distance=0.6cm,minimum height=6mm,minimum width=6mm,arw/.style={->,>=stealth'}]
  \node[rectangle,draw,rounded corners] at (0,0) (E) {Encoder};
   \node[rectangle] [right =0.5cm of E](X) {$X^n$};
  \node[rectangle] (Y) [right =1cm of X] {$\uY^n$};
  \node[rectangle] [above = 1cm of Y](Z){$\uZ^n$};
  \node[rectangle] (W) [left =0.4cm of E] {$W$};
  \node[rectangle,draw,rounded corners] (D) [right =2cm of Y] {Decoder};
  \node[rectangle,draw,rounded corners] (R) [right =0.4cm of Z] {Relay};
  \node[rectangle] (V) [right =0.4cm of R] {$V$};  
  \node[rectangle] (H) [right =0.4cm of D] {$\hat{W}$};
  \node[rectangle] (E1) [above =0.4cm of R] {$E$};  
  \node[rectangle] (E2) [below =0.4cm of D] {$E$}; 

  \draw [arw] (W) to node[midway,above]{} (E);
  \draw [arw] (E) to node[midway,above]{} (X);
  \draw [arw] (X) to node[midway,above]{} (Y);  
  \draw [arw] (X) to node[midway,right]{} (Z);
  \draw [arw] (Z) to node[midway,right]{} (R);
  \draw [arw] (R) to node[midway,right]{} (V);  
  \draw [arw] (V) to node[midway,right]{} (D); 
  \draw [arw] (Y) to node[midway,right]{} (D);  
  \draw [arw] (D) to node[midway,right]{} (H);  
  \draw [arw] (E1) to node[midway,right]{} (R);
  \draw [arw] (E2) to node[midway,right]{} (D);       
\end{tikzpicture}
\caption{Relay channel with side information}
\label{f_m2}
\end{figure}
\begin{prop}\label{prop1}
Suppose there exist encoder, relay, and decoder in Model~1 with parameters $R$, $R_0$ such that $P_e^{(n)}\le \epsilon$
for some $n$ and $\epsilon\in (0,1)$.
Then there exist relay, decoder in Model~2 (and with the same encoder in Model~1) such that
$P_e^{(n)}\le \epsilon$, $R$ and $R_0$ remain the same, and 
\begin{align}
\frac1{n}\mathbb{E}\left[\sum_{i=1}^n1_{X_i\in\mathcal{X}_{\rm bad}}\right]
\le
\Delta(C(\infty)-R+\mu(\epsilon))
\label{e_93}
\end{align}
where $X^n$ is the output of the encoder,
and $\Delta(\cdot)$ is the function defined in Proposition~\ref{prop_delta}.
\end{prop}
\begin{proof}
The relay can simulate $Z^n$ based on $E$ and $\uZ^n$.
Indeed, if $i$ is such that $E_i\in \mathcal{X}_{\rm bad}$ then generate $Z_i\sim P_{Z\vert X\uZ}(\cdot\vert E_i,\uZ_i)$ (where $P_{Z\vert X\uZ}$ is induced by $P_{Z\vert X}$ and the mapping $z\mapsto \uz$);
if $E_i=*$ then generate $Z_i$ with probability $\frac{Q_Z(z)}{Q_{\uZ}(\uz)}$ for each $z$ in the equivalent class $\uz$.
The simulation uses independent randomness for different $i$.
Then from the definition of the equivalent class we see that the distribution of $Z^n$ conditioned on $X^n$ follows the stationary memoryless channel with per-letter distribution $P_{Z\vert X}$.
Then with $Z^n$ available, $V$ can be computed using the relay coder in Model~1.
Similarly, $Y^n$ can be simulated based on $E$ and $\uY^n$,
and $\hat{W}$ is computed using the decoder in Model~1.
The joint distribution of $(W,Y^n,Z^n,V,\hat{W})$ is then the same as in Model~1, hence $P_e^{(n)}\le \epsilon$ is guaranteed.
Note that in the above, $V$ is computed stochastically from $E$ and $\uZ^n$ (i.e., the rule is given by a conditional distribution $P_{V\vert \uZ^n,E}$), 
and $\hat{W}$ is computed stochastically from $E$, $V$ and $\uY^n$. 
However, they can be converted to deterministic relay coder and decoder without increasing $P_e^{(n)}$, by outputting the $V$ and $\hat{W}$ with the least conditional error probability instead.

It remains to show \eqref{e_93}.
Observe that 
\begin{align}
n[C(\infty)-R+\mu(\epsilon)]
&\ge
D(P_{Y^nZ^n}\|Q_{Y^nZ^n})
\label{e94}
\\
&\ge \sum_{i=1}^nD(P_{Y_iZ_i}\|Q_{YZ})
\label{e_111}
\\
&\ge nD\left(\frac1{n}\sum_{i=1}^nP_{Y_iZ_i}\|Q_{YZ}\right)
\\
&= nD\left(\frac1{n}\sum_{i=1}^n
\sum_xP_{YZ\vert X=x}P_{X_i}(x)\|Q_{YZ}\right)
\\
&=nD\left(
\sum_xP_{YZ\vert X=x}\mathbb{E}[\widehat{P}_{X^n}(x)]\|Q_{YZ}\right),
\end{align}
where \eqref{e94} was shown in \eqref{e87};
\eqref{e_111} used $Q_{Y^nZ^n}=Q_{YZ}^{\otimes n}$;
 $\widehat{P}_{X^n}$ denotes the empirical distribution of $X^n$.
Therefore by Proposition~\ref{prop_delta},
\begin{align}
\sum_{x\in\mathcal{X}_{\rm bad}}\mathbb{E}[\widehat{P}_{X^n}(x)]
\le \Delta(C(\infty)-R+\mu(\epsilon)).
\end{align}
\end{proof}

Next, we shall sequentially introduce a number of functions. 
Roughly speaking, 
$H(\uZ\vert \uY)-R_0=O(C(\infty)-R)$ is implied if the first function is shown to vanish as its argument tends to zero.
The latter is in turn implied by the same vanishing property of the second function, and so on, which eventually leads to a form tractable by geometric tools.

\subsection{Function $g_n$}
\label{sec_gn}
Fix $P_{Y\vert X}$ and $P_{Z\vert X}$. 
For any $t\in[0,\infty)$, 
define 
\begin{align}
g_n(t):=\max\frac1{n}\mathbb{E}_P[\imath_{Q_{\underline{Z^n}\vert \underline{Y}^nVE}}]
\end{align}
where the max is over Model~2 source encoder, relay encoder, and decoder satisfying:
\begin{align}
\frac1{n}[\mathbb{E}_P[\imath_{Q_{\underline{Z^n}\vert \underline{Y}^nVE}}]-H(\underline{Z}^n\vert X^nV)]
\le t.
\label{e154}
\end{align}
The notations in the above definition is explained as follows, from which it will be seen that the left side of \eqref{e154} equals $\frac1{n}D(P_{\uZ^n\vert X^nV}\|Q_{\uZ^n\vert \uY^nVE}\vert P_{X^n\uY^nV})$ which is nonnegative.
\begin{itemize}
\item The joint distribution $P_{\uY^n\uZ^nX^nEV}$ is defined by the $P_{X^n}$ determined by the encoder, the channel $P_{\uY^n\uZ^n\vert X}=P_{\uY^n\vert X^n}P_{\uZ^n\vert X^n}$, $E=E(X^n)$, and $V=V(\uZ^n,E)$.
\item $H(Z^n\vert X^n,V)$ is defined for the distribution $P_{\uY^n\uZ^nX^nEV}$. 
Since $E$ is a function of $X^n$, we have $H(Z^n\vert X^n,V)=H(Z^n\vert X^n,V,E)$.
\item We pick an arbitrary fully supported $Q_E$ and define $Q_{\uY^n\uZ^nEV}$ by setting $$Q_{\uY^n\uZ^nE}=Q_EQ_{\uY^n\uZ^n}$$ where $Q_{\uY^n\uZ^n}$ is the capacity-achieving output distribution and $V=V(\uZ^n,E)$.
Note that $X^n$ will never appear under the  distribution $Q$.
Also later $Q$ will only be used when $E$ is been conditioned on, and so the choice of $Q_E$ will not matter since it does not affect the conditional distribution.
\item The information density $\imath_{Q_{\uZ^n\vert \uY^n,V,E}}$ is a function on $\mathcal{\uZ}^n\times \mathcal{\uY}^n\times \mathcal{V}\times \mathcal{E}$, defined by $\imath_{Q_{\uZ^n\vert \uY^n,V,E}}(\uz^n\vert \uy^n,v,e)=\log\frac1{Q_{\uZ^n\vert \uY^n,V,E}(\uz^n\vert \uy^n,v,e)}$.
\item $\mathbb{E}$ means taking the expectation of a function whereby the arguments are random variables following the distribution $P$.
\end{itemize}

\begin{lem}
\label{lem6}
Suppose that for parameters $n$, $R$, $R_0$ and $\epsilon$, Model~1 can achieve error $P_e^{(n)}\le \epsilon$.
Then
\begin{align}
H_{Q}(\uZ\vert \uY)-R_0
\le 
g_n\left(t\right)+t+\alpha\sqrt{\frac{\lambda}{2}}.
\end{align}
Here, $\alpha:=\max\abs{\imath_{Q_{\uZ\vert \uY}}(\uz\vert \uy)} <\infty$ where the max is over $(\uy,\uz)$ satisfying $\max_xP_{\uY\vert X=x}(\uy)P_{\uZ\vert X=x}(\uz)>0$\footnote{As in \eqref{e94}, the capacity-achieving output distribution $Q_{\uY\uZ}$ is fully supported  on $\{(\uy,\uz)\colon\max_xP_{\uY\vert X=x}(\uy)P_{\uZ\vert X=x}(\uz)>0\}$, hence $\alpha$ is finite.}.
$\mu(\cdot)$ is defined as in \eqref{e_mu}. 
$\Delta$ is defined in \eqref{e_Delta}.
$H_Q(\uZ\vert \uY)$ denotes the conditional entropy under $Q_{\uY\uZ}$.
Moreover $\lambda:=C(\infty)-R+\mu(\epsilon)$ and $t:=\mu(\Delta(\lambda))+\lambda+\frac{\log n}{n}$.
\end{lem}
\begin{proof}
Let the random variable $E$ be as defined in Model~2.
By Proposition~\ref{prop1}, we find a Model~2 scheme where $P_e^{(n)}\le \epsilon$ and with the same $R$ and $R_0$.
Then by Fano's inequality,
\begin{align}
nR-n\mu(\epsilon)&\le 
I(X^n;V,\underline{Y}^n,E)\\
&= I(X^n;E)
+I(X^n;\underline{Y}^n\vert E) 
+ I(X^n;V\vert \underline{Y}^n,E).
\label{e113}
\end{align}
The first term in \eqref{e113} is upper bounded by 
\begin{align}
H(E)&\le H(E\vert L)+H(L)
\\
&\le \mathbb{E}\left[\log\left(
\exp(n{\rm H}(\tfrac{L}{n}))
\cdot\vert \mathcal{X}\vert ^L
\right)\right]
+\log n
\\
&=n\mathbb{E}[{\rm H}(\tfrac{L}{n})]
+\mathbb{E}[L]\log\abs{\mathcal{X}}
+\log n
\\
&\le n{\rm H}(\tfrac{\mathbb{E}[L]}{n})
+\mathbb{E}[L]\log\abs{\mathcal{X}}
+\log n
\label{e_conv}
\\
&= n\mu(\tfrac{\mathbb{E}[L]}{n})+\log n
\\
&\le n\mu(\Delta(\lambda))
+\log n
\label{e119}
\end{align}
where we defined $L$ as the number of coordinates of $E$ in $\mathcal{X}_{\rm bad}; $\eqref{e_conv} follows from the concavity of the binary entropy function;
and \eqref{e119} follows from \eqref{e_93}.

The second term in \eqref{e113} is bounded by
\begin{align}
I(X^n;\uY^n\vert E) 
&\le I(X^n;\uY^n)
\\
&=D(P_{\underline{Y}^n\vert X^n}\|P_{\underline{Y}^n}\vert P_{X^n})
\label{e127}
\\
&\le D(P_{\underline{Y}^n\vert X^n}\|Q_{\underline{Y}^n}\vert P_{X^n})
\end{align}
where \eqref{e127} follows since $E(X^n)-X^n-\uY^n$ is a Markov chain.

The third term in \eqref{e113} is bounded as 
\begin{align}
I(X^n;V\vert \underline{Y}^n,E)
&=
D(P_{V\vert X^n}\| P_{V\vert \underline{Y}^nE}\vert P_{X^n\underline{Y}^n})
\label{e130}
\\
&\le D(P_{V\vert X^n}\| Q_{V\vert \underline{Y}^nE}\|P_{X^n\underline{Y}^n})
\\
&=D(P_{\underline{Z}^n\vert X^n}\| Q_{\underline{Z}^n\vert \underline{Y}^n}\|P_{X^n\underline{Y}^n})-
[\mathbb{E}_P[\imath_{Q_{\underline{Z^n}\vert \underline{Y}^nVE}}]-H(\underline{Z}^n\vert X^nV)],
\label{e132}
\end{align}
where \eqref{e130} follows since $P_{V\vert X^nE\uY^n}=P_{V\vert X^n}$,
and \eqref{e132} used the fact that $Q_{\underline{Z^n}\vert \underline{Y}^nE}=Q_{\underline{Z^n}\vert \underline{Y}^n}$ and that $V$ is a function of $E$ and $\uZ^n$ given the relay encoder.
Combining the bounds on the three terms in \eqref{e113} and rearranging, we obtain
\begin{align}
\mathbb{E}_P[\imath_{Q_{\underline{Z^n}\vert \underline{Y}^nVE}}]-H(\underline{Z}^n\vert X^nV)
\le 
nt
\end{align}
where we have chosen
\begin{align}
t:=\mu(\Delta(\lambda))
+\lambda+\frac{\log n}{n}.
\end{align}
By the definition of $g_n$ we must have 
\begin{align}
\frac{1}{n}\mathbb{E}_P[\imath_{Q_{\uZ^n\vert \uY^nVE}}]\le g_n(t). 
\label{e_135}
\end{align}
To finish, we can calculate the left side as
\begin{align}
&\quad\frac{1}{n}\mathbb{E}_P[\imath_{Q_{\uZ^n\vert \uY^nVE}}]
\nonumber\\
&=\frac{1}{n}
\mathbb{E}_P[\imath_{Q_{\uZ^n\vert \uY^n}}-\imath_{Q_{V\vert \uY^nE}}]
\label{e135}
\\
&\ge
\frac{1}{n}
\mathbb{E}_Q[\imath_{Q_{\uZ^n\vert \uY^n}}]
-\alpha \|P_{\uY_I\uZ_I}-Q_{\uY\uZ}\|_{TV}
-
\frac{1}{n}
\mathbb{E}_P[\imath_{Q_{V\vert \uY^nE}}]
\\
&\ge H_Q(\uZ\vert \uY)
-\alpha\sqrt{\frac{\lambda}{2}}
-
\frac{1}{n}
\mathbb{E}_P[\imath_{P_{V\vert \uY^nE}}]
-\frac{1}{n}D(P_{V\vert \uY^nE}\|Q_{V\vert \uY^nE}\vert P_{\uY^nE})
\label{e179}
\\
&\ge H_Q(\uZ\vert \uY)
-\alpha \sqrt{\frac{\lambda}{2}}
-R_0
-\lambda-\frac1{n}I(X^n;E),
\label{e_last179}
\\
&\ge H_Q(\uZ\vert \uY)
-\alpha \sqrt{\frac{\lambda}{2}}
-R_0
-t,
\label{e140}
\end{align}
where 
\begin{itemize}
\item \eqref{e135} used the fact that $Q_{\uZ^n\vert \uY^nE}=Q_{\uZ^n\vert \uY^n}$;
\item $\|\cdot\|_{TV}$ denotes the total variation distance; 
\item In \eqref{e179} we defined $I$ as an equiprobable random variable in $\{1,\dots,n\}$ independent of $(\uY^n,\uZ^n)$. 
\eqref{e179} follows by applying Pinsker's inequality to $D(P_{\uY_I\uZ_I}\|Q_{\uY\uZ})\le
\frac1{n}\sum_{i=1}^nD(P_{\uY_i\uZ_i}\|Q_{\uY\uZ})
\le \frac1{n}D(P_{\uY^n\uZ^n}\|Q_{\uY^n\uZ^n})
\le C(\infty)-R+\mu(\epsilon)$, where the last step follows from  \eqref{e87}.
\item To see \eqref{e_last179},
we first note $\mathbb{E}_P[\imath_{P_{V\vert \uY^nE}}]=H(V\vert \uY^nE)\le nR_0$.
The other term is bounded as
\begin{align}
D(P_{V\vert \uY^nE}\|Q_{V\vert \uY^nE}\vert P_{\uY^nE})
&\le 
D(P_{\uZ^n\vert \uY^nE}\|Q_{\uZ^n\vert \uY^n}\vert P_{\uY^nE})
\\
&= D(P_{\uZ^n\uY^n\vert E}\|Q_{\uZ^n\vert \uY^n}P_{\uY^n\vert E}\vert P_E)
\\
&\le D(P_{\uZ^n\uY^n\vert E}\|Q_{\uZ^n\vert \uY^n}Q_{\uY^n}\vert P_E)
\\
&\le 
n\lambda+I(X^n;E).\label{e143}
\end{align}
To justify \eqref{e143}, we use the similar argument as \eqref{e87}. That is,
\begin{align}
nR&\le I(X^n;\uZ^n\uY^nE)+n\mu(\epsilon)
\\
&
=I(X^n;\uZ^n\uY^n\vert E)+n\mu(\epsilon)+I(X^n;E)
\\
&=D(P_{\uY^n\uZ^n\vert X^n}\|Q_{\uY^n\uZ^n}\vert P_{X^n})-D(P_{\uY^n\uZ^n\vert E}\|Q_{\uY^n\uZ^n}\vert P_E)\nonumber\\&\quad+n\mu(\epsilon)+I(X^n;E)
\\
&\le nC(\infty)-D(P_{\uY^n\uZ^n\vert E}\|Q_{\uY^n\uZ^n}\vert P_E)+n\mu(\epsilon)+I(X^n;E)
\end{align}
which is equivalent to \eqref{e143} upon rearrangement.
\item \eqref{e140} follows from \eqref{e119}.
\end{itemize}
Combining \eqref{e_135} and \eqref{e_last179}, we have
\begin{align}
g_n(t)+ t+\alpha \sqrt{\frac{\lambda}{2}}
\ge H_Q(\uZ\vert \uY)-R_0.
\end{align}
\end{proof}

\subsection{Function $f_n(t)$}
Fix $P_{Y\vert X}$ and $P_{Z\vert X}$. 
For any $t\in[0,\infty)$, 
define 
\begin{align}
f_n(t):=\max\frac1{n}H(\underline{Z}^n\vert X^n,V)
\end{align}
where the max is over Model~2 source encoder, relay encoder, and decoder satisfying:
\begin{align}
\frac1{n}[\mathbb{E}_P[\imath_{Q_{\underline{Z^n}\vert \underline{Y}^nVE}}]-H(\underline{Z}^n\vert X^nV)]\le t.
\label{e181}
\end{align}
As before, the left side of \eqref{e181} can be expressed as the relative entropy $\frac1{n}D(P_{\underline{Z}^n\vert X^nV}\|Q_{\underline{Z}^n\vert \underline{Y}^nVE}\|P_{X^nV\underline{Y}^n})$ which is nonnegative.

\begin{lem}
\label{lem7}
\begin{align}
g_n(t)\le \max\{2t,2f_n(t)\}
\end{align}
\end{lem}

\begin{proof}
Fix Model~2 source encoder, relay coder, and decoder achieving the max in the definition of $g_n(t)$. 
Then 
\begin{align}
\mathbb{E}_P[\imath_{Q_{\underline{Z^n}\vert \underline{Y}^nVE}}]-H(\underline{Z}^n\vert X^nV)\le nt.
\label{e139}
\end{align}
holds.
\begin{itemize}
\item If $\mathbb{E}_P[\imath_{Q_{\underline{Z^n}\vert \underline{Y}^nVE}}]\ge 2H(\underline{Z}^n\vert X^nV)$,
then cancelling $H(\underline{Z}^n\vert X^nV)$ in \eqref{e139}, we have 
\begin{align}
\mathbb{E}_P[\imath_{Q_{\underline{Z^n}\vert \underline{Y}^nVE}}]
\le 2nt
\end{align}
which shows that $g_n(t)\le 2nt$.
\item If $\mathbb{E}_P[\imath_{Q_{\underline{Z^n}\vert \underline{Y}^nVE}}]\le 2H(\underline{Z}^n\vert X^nV)$,
we have 
\begin{align}
f_n(t)\ge H(\underline{Z}^n\vert X^nV)
\ge \frac1{2}\mathbb{E}_P[\imath_{Q_{\underline{Z^n}\vert \underline{Y}^nVE}}]
=\frac1{2}g_n(t).
\end{align}
\end{itemize}
\end{proof}

\subsection{Function $\bar{f}_n(t)$}
For given $x^n$, the \emph{condition type} of a sequence $z^n$ is the conditional law of $Z$ given $X$ in the empirical distribution $\widehat{P}_{x^nz^n}$.
Let 
$
\mathcal{P}_{x^n}
$
be the set of all possible such conditional types.
If $p\in \mathcal{P}_{x^n}$, then 
denote by 
$
\mathcal{T}_{x^n}(p)
$
the set of $z^n$ sequences with conditional type $p$.
We use 
$
T_{x^n}(z^n)
$
to denote the conditional type the sequence $z^n$ (given $x^n$).

Fix $P_{Y\vert X}$ and $P_{Z\vert X}$. 
For any $t\in[0,\infty)$, 
define 
\begin{align}
\bar{f}_n(t):=\max\frac1{n}H(\underline{Z}^n\vert X^n=x^n,V=v,T_{x^n}(\underline{Z}^n)=p)
\end{align}
where the max is over Model~2 source encoder, relay encoder, decoder, $x^n$ and $v$ and $p$ satisfying:
\begin{itemize}
\item $P_{X^n,V,T_{x^n}(\underline{Z}^n)}(x^n,v,p)>0$.
\item 
$\frac1{n}D(P_{\underline{Z}^n\vert X^n=x^n,V=v,T_{x^n}(\underline{Z}^n)=p}\|Q_{\underline{Z}^n\vert \underline{Y}^n,V=v,E=E(x^n),T_{x^n}(\underline{Z}^n)=p}\vert P_{\underline{Y}^n\vert X^n=x^n})\le t$;
\end{itemize}
\begin{lem}
\label{lem8}
For any $t>0$,
\begin{align}
f_n(t)\le 
\bar{f}_n(t^{4/5})+t^{1/5}\log\abs{\mathcal{Z}}+O(\frac1{n}\log n)
\end{align}
where the $O(\frac1{n}\log n)$ term may be chosen as $\frac1{n}\abs{\mathcal{Z}}\abs{\mathcal{X}}\log(n+\abs{\mathcal{Z}}-1)$.
\end{lem}

\begin{proof}
Choose Model~2 source encoder, relay encoder, and decoder achieving max in the definition of $f_n(t)$.
Then we have 
\begin{align}
D(P_{\underline{Z}^n\vert X^nV}\|Q_{\underline{Z}^n\vert \underline{Y}^nVE}\|P_{X^nV\underline{Y}^n})
\le nt.
\end{align}
Since 
$T_{X^n}(\uZ^n)$ is a function of $\underline{Z}^n$ conditioned on $X^n$, 
we have
\begin{align}
&\quad D(P_{\underline{Z}^n\vert X^nV}\|Q_{\underline{Z}^n\vert \underline{Y}^nVE}\|P_{X^nV\underline{Y}^n})
\nonumber\\
&=
D(P_{\underline{Z}^nT_{X^n}(\underline{Z}^n)\vert X^nV}\|Q_{\underline{Z}^nT_{X^n}(\underline{Z}^n)\vert \underline{Y}^nVE}\|P_{X^nV\underline{Y}^n})
\\
&\ge 
D(P_{\underline{Z}^n\vert X^nVT_{X^n}(\underline{Z}^n)}\|Q_{\underline{Z}^n\vert \underline{Y}^nVET_{X^n}(\underline{Z}^n)}\|P_{X^nV\underline{Y}^nT_{X^n}(\underline{Z}^n)}).
\end{align}

By the Markov inequality and the union bound, there exists a set $\mathcal{G}$ of $(x^n,v,p)$ such that 
$
P_{X^nVT_{X^n}(\uZ^n)}(\mathcal{G})\ge 1-t^{1/5}
$,
and for any $(x^n,v,p)\subseteq \mathcal{G}$,
\begin{align}
P_{X^nVT_{X^n}(\uZ^n)}(x^n,v,p)&>0,
\\
D(P_{\underline{Z}^n\vert X^n=x^n,V=v,T_{X^n}(\underline{Z}^n)=p}\|Q_{\underline{Z}^n\vert \underline{Y}^n,V=v,E=E(x^n),T_{X^n}(\underline{Z}^n)=p}\vert P_{\underline{Y}^n\vert X^n=x^n})&\le nt^{4/5}.
\end{align}
Note that $P_{X^n,V,T_{x^n}(\underline{Z}^n)}(x^n,v,p)>0$ ensures that the conditional probability $P_{\underline{Z}^n\vert X^n=x^n,V=v,T_{x^n}(\underline{Z}^n)=p}$ is well-defined. 
Then $Q_{\underline{Z}^n\vert \underline{Y}^n=\uy^n,V=v,E=E(x^n),T_{x^n}(\underline{Z}^n)=p}$ is well-defined when $P_{\uY^n\vert X^n=x^n}(\uy^n)>0$,
because the support of $Q_{\underline{Y}^n\underline{Z}^n}$ contains the support of $P_{\underline{Y}^n\underline{Z}^n\vert X^n=x^n}$ for any $x^n$.
Then,
\begin{align}
nf_n(t)&=H(\underline{Z}^n\vert X^n,V)
\\
&=H(\underline{Z}^n,T_{X^n}(\underline{Z}^n)\vert X^n,V)
\\
&=H(\underline{Z}^n\vert X^n,V,T_{X^n}(\underline{Z}^n))
+H(T_{X^n}(\underline{Z}^n)\vert X^n,V)
\\
&=H(\underline{Z}^n\vert X^n,V,T_{X^n}(\underline{Z}^n))+O(\log n)
\\
&=\sum_{(x^n,v,p)\in\mathcal{G}} H(\underline{Z}^n\vert X^n=x^n,V=v,T_{X^n}(\underline{Z}^n)=p)P_{X^nVT_{X^n}(\underline{Z}^n)}(x^n,v,p)
\nonumber\\
&\quad+
\sum_{(x^n,v,p)\in\mathcal{G}^c} H(\underline{Z}^n\vert X^n=x^n,V=v,T_{X^n}(\underline{Z}^n)=p)P_{X^nVT_{X^n}(\underline{Z}^n)}(x^n,v,p)
\nonumber\\
&\quad+O(\log n)
\\
&\le n\bar{f}_n(t^{4/5})+nt^{1/5}\log\abs{\mathcal{Z}}+O(\log n).
\end{align}
\end{proof}

\subsection{Function $h_n$}\label{sec_h}
\begin{defn}
Let $d_1$, $d_2$,\dots, $d_n$ be a sequence of nonnegative integers. 
For each $i=1,\dots,n$, let $\mathcal{C}_i\subseteq\mathbb{R}^{d_i}$ be a bounded, symmetric convex set containing $0$ in its interior, and let $\kappa>0$.
We say a set $\mathcal{A}_n\subseteq\mathbb{R}^{d_1}\times\dots\times \mathbb{R}^{d_n}$ is $\kappa$-differentiated (with respect to $\mathcal{C}_1$,\dots,$\mathcal{C}_n$) if for any $(z_1,\dots,z_n),(z_1',\dots,z_n')\in\mathcal{A}_n$, we have
\begin{align}
\sum_{i=1}^n\|z_i-z_i'\|_{(\mathcal{C}_i)^{\circ}}
\ge\kappa\sum_{i=1}^n1_{z_i\neq z_i'}
\end{align}
\end{defn}
Given $\kappa>0$, $t\ge0$, and nonnegative integers $n$, $d$, 
define $h_n(t)$ to be the smallest number such that the following holds:
For any sequence of integers $d_1,\dots,d_n\in\{0,\dots,d\}$
and $\mathcal{A}_n\subseteq \mathbb{R}^{d_1+\dots+d_n}$ satisfying
\begin{itemize} 
\item $\mathcal{A}_n$ is $\kappa$-differentiated with respect to $\mathcal{C}_i=[-1,1]^{d_i}$, $i=1,\dots,n$;
\item  $\abs{\{z_i\colon \exists z^n\in\mathcal{A}_n\}}\le d+1$, $i=1,\dots,n$;
\item $E_{\bar{Y}^n}[\rho(\bar{Y}^n)]\le t$, 
where $\bar{Y}^n$ is equiprobable on $\mathcal{C}^n:=\mathcal{C}_1\times\dots\times\mathcal{C}_n$ and
\begin{align}
\rho(y^n):=\log\left(\mathbb{E}[\exp\left< y^n,Z^n\right>]\right)
-\mathbb{E}[\left< y^n,Z^n\right>]
\label{e_rho}
\end{align}
with $Z^n$ equiprobable on $\mathcal{A}_n$,
\end{itemize}
we have 
\begin{align}
\frac1{n}\log\abs{\mathcal{A}_n}\le h_n(t).
\end{align}

Note that since $\bar{Y}^n$ is equiprobable on $\mathcal{C}^n$ whereas the $\kappa$-differentiated-ness of $\mathcal{A}_n$ is gauged by $(\mathcal{C}^n)^{\circ}$,
we have the following invariance property as the inner products in \eqref{e_rho} shall be preserved.
\begin{prop}\label{prop_invar}
The definition of $h_n(c)$ does not change if each $\mathcal{C}_i=[-1,1]^{d_i}$ is replaced with an arbitrary nondegenerate linear transform of $[-1,1]^{d_i}$.
 \end{prop}
 
We now make the connection between $h_n$ and the relay channel problem.
For each $x$ and $\uy$, $\uz$ satisfying $P_{\uY\vert X}(\uy\vert x)>0$, $P_{\uZ\vert X}(\uz\vert x)>0$,
define 
\begin{align}
K_x(\uy,\uz):=-\log \frac{P_{\uZ\vert X=x}(\uz)}{Q_{\uZ\vert \uY=\uy}(\uz)}
\label{e172}
\end{align}
which is in $(-\infty,\infty)$.
Define
\begin{align}
\tilde{K}_x(\uy,\uz)
=K_x(\uy,\uz)-\sum_{\uy'}P_{\uY\vert X=x}(\uy')K_x(\uy',\uz).
\label{e173}
\end{align}
Then we have 
\begin{align}
\mathbb{E}_{P_{\uY\vert X=x}}[\tilde{K}_x(\uY,\uz)]=0.\label{e174}
\end{align}
Moreover $K_x(\uy,\uz)-\tilde{K}_x(\uy,\uz)$ depends only on $x$ and $\uz$.
Thus for any $x^n$, $\uy^n$, $p$ and $\uz^n\in\mathcal{T}_{x^n}(p)$ such that $P_{\uY^n\vert X^n}(\uy^n\vert x^n)>0$ and $P_{\uZ^n\vert X^n}(\uz^n\vert x^n)>0$, we have that
\begin{align}
\sum_{i=1}^nK_{x_i}(\uy_i,\uz_i)
-\sum_{i=1}^n\tilde{K}_{x_i}(\uy_i,\uz_i)
=c_{x^n,p}
\label{e209}
\end{align}
is a constant depending only on $x^n$ and $p$ (not on $\uy^n$ and the particular choice of $\uz^n\in\mathcal{T}_{x^n}(p)$).

For each $x\in\mathcal{X}$ define $\mathcal{\uY}_x:=\{\uy\in\mathcal{\uY}\colon P_{\uY\vert X=x}(\uy)>0\}$,
and let $d(x):=\abs{\mathcal{\uY}_x}-1$.
Define $\mathcal{\uZ}_x$ similarly.
Next we will define representations
\begin{align}
\phi_x\colon \mathcal{\uY}_x\to \mathbb{R}^{d(x)}
\\
\psi_x \colon \mathcal{\uZ}_x\to \mathbb{R}^{d(x)}
\end{align}
so that 
\begin{align}
\left<\phi_x(\uy), \psi_x(\uz)\right>
=\tilde{K}_x(\uy,\uz)
\label{e_rep}
\end{align}
for any $x$, $\uy\in\mathcal{\uY}_x$, $\uz\in\mathcal{\uZ}_x$.
Clearly such representations are well-defined only up to a linear transform. 
It will be seen that the particular choice of the representation will not make any difference due to the invariance property (Proposition~\ref{prop_invar}), 
but for concreteness we can make the following choice:
call elements in $\mathcal{\uY}_x$ as $\uy^0$, $\uy^1$,\dots, $\uy^{d(x)}$.
Set $\phi_x$ to map these elements to the following vectors in $\mathbb{R}^d$:
\begin{align}
-\sum_{i=1}^{d(x)}\frac{P_{\uY\vert X=x}(\uy^i)}{P_{\uY\vert X=x}(\uy^0)}e_i,e_1,\dots,e_{d(x)}
\label{e_simplex}
\end{align}
where $e_1,\dots,e_{d(x)}$ are the canonical basis vectors for $\mathbb{R}^{d(x)}$.
Thus 
\begin{align}
\sum_{\uy\in\mathcal{\uY}_x}\phi_x(\uy)P_{\uY\vert X=x}(\uy)=0.
\end{align}
Define $\psi_x$ by
\begin{align}
\psi_x(\uz):=(\tilde{K}_x(\uy^i,\uz))_{i=1}^{d(x)},\quad
\forall x\in\mathcal{X},\,\uz\in \mathcal{Z}_x.
\end{align}
Then \eqref{e_rep} is clearly satisfied for $\uy=\uy^i$, $i\neq 0$. 
For $i=0$, we have
\begin{align}
\tilde{K}_x(\uy^0,\uz)
&=-\sum_{i=1}^{d(x)}\frac{P_{\uY\vert X=x}(\uy^i)}{P_{\uY\vert X=x}(\uy^0)}\tilde{K}_x(\uy^i,\uz)
\label{e182}
\\
&=\left<-\sum_{i=1}^{d(x)}\frac{P_{\uY\vert X=x}(\uy^i)}{P_{\uY\vert X=x}(\uy^0)}\phi_x(\uy
^i),\psi_x(\uz)\right>
\\
&=\left<\phi_x(\uy^0),\psi_x(\uz)\right>
\end{align}
where \eqref{e182} used \eqref{e174};
therefore \eqref{e_rep} remains true.

We now show that for any $x\in\mathcal{X}$, the map $\psi_x$ is injective on $\mathcal{\uZ}_x$ under the symmetry assumption $P_{Y\vert X}=P_{Z\vert X}$.
From the definition \eqref{e173} we can see that 
\begin{align}
\tilde{K}_x(\uy,\uz)=\log \frac{Q_{\uY\uZ}(\uy,\uz)}{Q_{\uY}(\uy)}-a(x,\uz),\quad \forall \uy\in\mathcal{\uY}_x,\,\uz\in\mathcal{\uZ}_x
\label{e_kt}
\end{align}
for some function $a$ of $(x,\uz)$.
Pick arbitrary $\uz\neq \uz'$ in $\mathcal{\uZ}_x$.
The following argument uses the symmetry assumption $P_{Y\vert X}=P_{Z\vert X}$:
Observe that the matrix 
\begin{align}
\left(
\begin{array}{cc}
Q_{\uY\uZ}(\uz,\uz)  &   Q_{\uY\uZ}(\uz,\uz')  \\
Q_{\uY\uZ}(\uz',\uz)  &   Q_{\uY\uZ}(\uz',\uz')   
\end{array}
\right)
\label{e_full}
\end{align}
is full-rank.
Indeed, if otherwise, we must have $(P_{\uZ\vert X}(\uz\vert x),P_{\uZ\vert X}(\uz'\vert x))$ all being a multiple of some vector $v\in\mathbb{R}^2$ for all $x$ in the support of a capacity-achieving input distribution (to see this, consider the fact that if $u,v\in\mathbb{R}^2$ are vectors which are linearly independent, then $uu^{\top}+vv^{\top}$ must be strictly positive definite).
Integrating over $x$, we find that $(Q_{\uZ\uY}(\uz,\uy),Q_{\uZ\uY}(\uz',\uy))$ is a multiple of $v$ for each $\uy\in\mathcal{\uY}$.
Then $Q_{\uY\vert \uZ=\uz}=Q_{\uY\vert \uZ=\uz'}$, a contradiction (Proposition~\ref{prop7}).
Now we claim that 
$(\tilde{K}_x(\uy,\uz))_{\uy\in\mathcal{\uY}_x}\neq(\tilde{K}_x(\uy,\uz'))_{\uy\in\mathcal{\uY}_x}$,
hence the injectivity of $\psi_x$. 
Indeed, otherwise $(\tilde{K}_x(\uy,\uz))_{\uy\in\mathcal{\uY}_x}=(\tilde{K}_x(\uy,\uz'))_{\uy\in\mathcal{\uY}_x}$ implies that the vectors $(Q_{\uY\uZ}(\uy,\uz))_{\uy\in\mathcal{\uY}_x}$ and $(Q_{\uY\uZ}(\uy,\uz'))_{\uy\in\mathcal{\uY}_x}$ differ by a multiplicative constant
according to \eqref{e_kt}, which contradicts the fact that \eqref{e_full} is full-rank.

For each $x\in\mathcal{X}$, let $\mathcal{C}_x$ be an arbitrary nondegenerate linear transform of $[-1,1]^{d(x)}$ contained in the convex hull of \eqref{e_simplex},\footnote{In principle, we can maximize $\kappa$ by optimizing over $\mathcal{C}_x$ contained in the convex hull of \eqref{e_simplex}, although for our purpose of solving Cover's problem we can pick any $\mathcal{C}_x$.}
and then define 
\begin{align}
\kappa:=\min_{x\in\mathcal{X}}\,
\min_{\uz,\uz'\in\mathcal{\uZ}_x,\,\uz\neq \uz'}
\|\psi_x(\uz)-
\psi_x(\uz')\|_{(\mathcal{C}_x)^{\circ}}.
\label{e216}
\end{align}
Now by the injectivity of $\psi_x$ shown above we see that 
\begin{prop}
Assume symmetry $P_{Y\vert X}=P_{Z\vert X}$.
Then
$\kappa\in (0,\infty)$.
\end{prop}

\begin{lem}
\label{lem11}
Assume symmetry $P_{Y\vert X}=P_{Z\vert X}$.
Let $\kappa$ be defined as in \eqref{e216},
and $d:=\abs{\mathcal{Y}}-1$.
Then
$\bar{f}_n(t)\le h_n(t)$.
\end{lem}
\begin{proof}
Choose Model~2 source encoder, relay encoder, decoder, $x^n$ and $v$ and $p$, 
in the definition of $\bar{f}_n(t)$.
Consider arbitrary $\uz^n$ in
\begin{align}
\mathcal{A}_{x^n,v,p}:=\{\uz^n\colon V(\uz^n,E(x^n))=v,\,T_{x^n}(\uz^n)=p\}.
\end{align}
We will show that 
\begin{align}
&\quad\frac{P_{\uZ^n\vert X^n=x^n,V=v,T_{x^n}(\uZ^n)=p}(\uz^n)}
{Q_{\uZ^n\vert \uY^n=\uy^n,V=v,E=E(x^n),T_{x^n}(\uZ^n)=p}(\uz^n)}
\nonumber\\
&=c\,\frac{P_{\uZ^n\vert X^n}(\uz^n\vert x^n)}{Q_{\uZ^n\vert \uY^n}(\uz^n\vert \uy^n)}
\label{e218}
\end{align}
where $c>0$ is some constant depending only on $x^n$, $v$, $p$ and $\uy^n$ (not on $\uz^n$). 
Indeed,
\begin{align}
P_{\uZ^n\vert X^n=x^n,V=v,T_{x^n}(\uZ^n)=p}(\uz^n)
&=\frac{P_{\uZ^nVT_{x^n}(\uZ^n)\vert X^n=x^n}(\uz^n,v,p)}{P_{VT_{x^n}(\uZ^n)\vert X^n=x^n}(v,p)}
\\
&=\frac{P_{\uZ^n\vert X^n=x^n}(\uz^n)}{P_{VT_{x^n}(\uZ^n)\vert X^n=x^n}(v,p)},
\end{align}
whereas
\begin{align}
&\quad Q_{\uZ^n\vert \uY^n=\uy^n,V=v,E=E(x^n),T_{x^n}(\uZ^n)=p}(\uz^n)
\nonumber\\
&=\frac{Q_{\uZ^nVT_{x^n}(\uZ^n)\vert \uY^n=\uy^n,E=E(x^n)}(\uz^n,v,p)}{Q_{VT_{x^n}(\uZ^n)\vert \uY^n=\uy^n,E=E(x^n)}(v,p)}
\label{e_192}
\\
&=\frac{Q_{\uZ^nV\vert \uY^n=\uy^n,E=E(x^n)}(\uz^n,v)}{Q_{VT_{x^n}(\uZ^n)\vert \uY^n=\uy^n,E=E(x^n)}(v,p)}
\label{e_193}
\\
&=\frac{Q_{\uZ^n\vert \uY^n=\uy^n,E=E(x^n)}(\uz^n)Q_{V\vert \uZ^n=\uz^n,E=E(x^n),\uY^n=\uy^n}(v)}
{Q_{VT_{x^n}(\uZ^n)\vert \uY^n=\uy^n,E=E(x^n)}(v,p)}
\label{e_194}
\\
&=\frac{Q_{\uZ^n\vert \uY^n=\uy^n}(\uz^n)
Q_{V\vert \uZ^n=\uz^n,E=E(x^n),\uY^n=\uy^n}(v)}
{Q_{VT_{x^n}(\uZ^n)\vert \uY^n=\uy^n,E=E(x^n)}(v,p)},
\label{e_195}
\\
&=\frac{Q_{\uZ^n\vert \uY^n=\uy^n}(\uz^n)}{Q_{VT_{x^n}(\uZ^n)\vert \uY^n=\uy^n,E=E(x^n)}(v,p)},
\label{e_196}
\end{align}
where \eqref{e_195} follows since $E$ and $(\uY^n,\uZ^n)$ are independent in the definition of $Q$, and 
\eqref{e_196} follows since $\uz^n\in\mathcal{A}_{x^n,v,p}$ implies that $Q_{V\vert \uZ^n=\uz^n,E=E(x^n),\uY^n=\uy^n}(v)=1$.
Therefore \eqref{e218} is verified.
Now 
\begin{align}
t&\ge  D(P_{\uZ^n\vert X^n=x^n,V=v,T_{x^n}(\uZ^n)=p}\|
Q_{\uZ^n\vert \uY^n=\uy^n,V=v,E=E(x^n),T_{x^n}(\uZ^n)=p})
\nonumber\\
&=
\log\mathbb{E}[\exp(K_{x^n}(\uy^n,\cdot))]
-\mathbb{E}[K_{x^n}(\uy^n,\cdot)]
\label{e231}
\\
&=
\log\mathbb{E}[\exp(\tilde{K}_{x^n}(\uy^n,\cdot))]
-\mathbb{E}[\tilde{K}_{x^n}(\uy^n,\cdot)]
\label{e232}
\\
&=\log\mathbb{E}[\exp(\left<\phi_{x^n}(\uy^n),
\psi_{x^n}(\cdot)\right>)]
-\mathbb{E}[\left<\phi_{x^n}(\uy^n),
\psi_{x^n}(\cdot)\right>]
\label{e199}
\end{align}
where 
\begin{itemize}
\item the expectations are with respect to the distribution $P_{\uZ^n\vert VT_{x^n}(\uZ^n)X^n}(\cdot\vert v,p,x^n)$, which is the equiprobable distribution on $\mathcal{A}_{x^n,v,p}$.
\item \eqref{e231} follows from \eqref{e218},
where we defined $K_{x^n}(\uy^n,\uz^n):=\sum_{i=1}^nK_{x_i}(\uy_i,\uz_i)$, and define $\tilde{K}_{x^n}$ similarly
(recall the definitions of $K_x$ and $\tilde{K}_x$ in \eqref{e172} and \eqref{e173});
\item \eqref{e232} follows since under $P_{\uZ^n\vert VT_{x^n}(\uZ^n)X^n}(\cdot\vert v,p,x^n)$ 
we have $K_{x^n}(\uy^n,\uz^n)
-\tilde{K}_{x^n}(\uy^n,\uz^n)
=c_{x^n,p}$ with probability 1
as shown in \eqref{e209};
\item \eqref{e199} follows from \eqref{e_rep}, and we defined $\phi_{x^n}(\uy^n):=(\phi_{x_1}(\uy_1),\dots,\phi_{x_n}(\uy_n))$, and similarly for $\psi_{x^n}$.
\end{itemize}
Now setting $d_i:=d(x_i)$, $\mathcal{C}_i:=\mathcal{C}_{x_i}$ (where $\mathcal{C}_x$ was chosen near \eqref{e216}), $i=1,\dots,n$, and 
defining $\mathcal{A}_n$ as the image of $\mathcal{A}_{x^n,v,p}$ under $\psi_{x^n}\colon \uz^n\mapsto (\psi_{x_i}(\uz_i))_{i=1}^n$,
we see that $\mathcal{A}_n$ is $\kappa$-differentiated by \eqref{e216},
and 
\begin{align}
D(P_{\uZ^n\vert X^n=x^n,V=v,T_{x^n}(\uZ^n)=p}\|
Q_{\uZ^n\vert \uY^n=\uy^n,V=v,E=E(x^n),T_{x^n}(\uZ^n)=p})
=\rho(\phi_{x^n}(\uy^n))
\label{e200}
\end{align}
where $\rho$ is defined as in \eqref{e_rho}.
Now let $\bar{Y}_1,\dots,\bar{Y}_n$ be independent and $P_{\bar{Y}_i}$ be equiprobable on $\mathcal{C}_i$.
Since $\mathcal{C}_i$ is contained in the convex hull of \eqref{e_simplex} which is a simplex,
there exists a coupling of $P_{\phi_{x^n}(\uY^n)\vert X^n=x^n}$ and $P_{\bar{Y}^n}$ such that $\mathbb{E}[\phi_{x^n}(\uY^n)\vert \bar{Y}^n]=\bar{Y}^n$ (indeed, for each $i$ and $\bar{y}_i\in\mathcal{C}_i$ we can define $P_{\phi_{x_i}(\uY_i)\vert \bar{Y}_i=\bar{y}_i}$ as the unique distribution on the points in \eqref{e_simplex} under which the expectation equals $\bar{y}_i$; this couples $P_{\phi_{x_i}(\uY_i)\vert X=x_i}$ and $P_{\bar{Y}_i}$ because $\mathbb{E}[\phi_{x_i}(\uY_i)\vert X=x_i]=\mathbb{E}[\bar{Y}_i]=0$).
Then 
\begin{align}
t\ge\mathbb{E}[\rho(\phi_{x^n}(\uY^n))\vert X^n=x^n]
\ge
\mathbb{E}[\rho(\bar{Y}^n)] 
\end{align}
where the first inequality follows from \eqref{e200} and the definition of $\bar{f}_n(t)$, and the second inequality follows from Jensen's inequality since $\rho$ is convex.
Thus by the definition of $h_n(t)$ (noting Proposition~\ref{prop_invar}) we see that 
\begin{align}
H(\uZ^n\vert X^n=x^n,V=v,T_{x^n}(\uZ^n)=p)=\log\abs{\mathcal{A}_n}\le h_n(t).
\end{align}
It follows that $\bar{f}_n(t)\le h_n(t)$.
\end{proof}

\section{Sampling from a Polytope}\label{sec_sampling}
%
The goal of this section is to bound $h_n(t)$ defined in Section~\ref{sec_h}.
The basic idea is find a (generally nonproduct) subset $\mathcal{S}_n$ of the product set $\mathcal{C}^n$, so that a ``soft-max bound'' of $\mathcal{A}_n$ for $y^n$ in $\mathcal{C}^n$ (bound on $\rho(y^n)$) is translated to a ``hard-max bound'' of another set $\mathcal{B}_n$ for $y^n$ in $\mathcal{S}_n$.
Conceptually, we may draw an analogy to the familiar $\epsilon$-net argument in bounding the supremum of processes \cite{boucheron2004concentration},
whereby $\mathcal{S}_n$ can be thought of as a ``net'' of $\mathcal{C}^n$.
\subsection{A Construction via Hadamard Matrices}
\label{sec_had}
\begin{lem}\label{lem_sampling}
Let $d_1$,\dots,$d_n$ be nonnegative integers and let $\mathcal{C}_i$ be 
a nondegenerate linear transform of $[-1,1]^{d_i}$, $i=1,\dots,n$.
Put $N:=d_1+\dots+d_n$.
For any integer $t$ which is a power of $2$,
there exists $\mathcal{S}_n \subseteq \mathcal{C}^n:=\mathcal{C}_1\times \dots\times \mathcal{C}_n$ such that
\begin{align}
\abs{\mathcal{S}_n}=
(2t)^{\lfloor \frac{N}{t}\rfloor}
\cdot 2^{N-\lfloor \frac{N}{t}\rfloor t}
\end{align}
and 
\begin{align}
\frac1{\sqrt{t}}\,\mathcal{C}^n\subseteq \conv(\mathcal{S}_n).
\label{e_scontain}
\end{align}
\end{lem}
\begin{proof}
It is easy to see from the invariance under linear transforms that we can assume without loss of generality that $\mathcal{C}_i=[-1,1]^{d_i}$.
Then $\mathcal{C}^n$ is simply $[-1,1]^N$.
Now suppose that $t$ is a power of $2$.
Let $\mathcal{R}$ be the set of the following $2t$ points in $\mathbb{R}^t$:
\begin{align}
(0,0,\dots,\pm 1,\dots,0,0),
\end{align}
that is, points with exactly one nonzero coordinate which is equal to either 1 or $-1$.
The next step is to choose a $t\times t$ matrix $A_t$ with $\pm1$ entries so that points in $\mathcal{R}$ are transformed to points in $[-1,1]^t$.
Let $A_t$ be a Hadamard matrix of order $t$, which is possible because $t$ is a power of $2$. 
We have 
\begin{align}
A_t^{\top}A_t
&=
tI_t.
\label{e_iso1}
\end{align}
Now we choose $\mathcal{S}_n\subseteq \mathcal{C}^n=[-1,1]^N$ as
\begin{align}
\underbrace{
(A_t\mathcal{R})
\oplus
\dots
\oplus
(A_t\mathcal{R})}_\text{$\lfloor \frac{N}{t}\rfloor$ times}
\oplus
\underbrace{
\{1,-1\}\oplus\dots\oplus\{1,-1\}}_\text{$N-\lfloor \frac{N}{t}\rfloor t$ times}.
\end{align}
Thus
\begin{align}
\abs{\mathcal{S}_n}=(2t)^{\lfloor \frac{N}{t}\rfloor}
\cdot 2^{(N-\lfloor \frac{N}{t}\rfloor t)}
\end{align}
We now show that 
\begin{align}
t^{-\frac1{2}}[-1,1]^t\subseteq\conv(A_t\mathcal{R}),
\label{e_include}
\end{align}
which will imply the claim of the lemma.
Pick an arbitrary $x=(x_1,\dots,x_t)\in\mathbb{R}^t$
 such that 
\begin{align}
\sum_{i=1}^t\abs{x_i}=1.
\label{e_boundary}
\end{align}
Note that \eqref{e_boundary} is equivalent to $x$ being on the boundary of $\conv(\mathcal{R})$,
and hence is equivalent to $A_tx$ being on the boundary of $\conv(A_t\mathcal{R})$.
We will show that
\begin{align}
\sup_{1\le j\le t}\abs{(A_tx)_j}
\ge t^{-\frac1{2}}
\label{e_norm}
\end{align}
where $(A_tx)_j$ denotes the $j$-th coordinate of $A_tx$, which is equivalent to $A_tx$ being outside $t^{-\frac1{2}}(-1,1)^t$.
Since both $[-1,1]^t$ and $\conv(A_t\mathcal{R})$ are non-degenerate convex sets containing zero, this would imply \eqref{e_include}.

To prove \eqref{e_norm}, note that 
\begin{align}
\|x\|_2
&\ge
t^{-\frac1{2}}\|x\|_1
=t^{-\frac1{2}}.
\end{align}
Therefore, 
\begin{align}
\sup_{1\le j\le t}\abs{(A_tx)_j}
&\ge \frac1{\sqrt{t}}\|A_tx\|_2
\\
&=\|x\|_2 \label{e_iso}
\\
&\ge t^{-\frac1{2}}
\end{align}
where \eqref{e_iso} used \eqref{e_iso1}.
Hence \eqref{e_norm} is established, and the claim of the lemma is true.
\end{proof}
\begin{rem}
The explicit construction of $\mathcal{S}_n$ in Lemma~\ref{lem_sampling} is based on Hadamard matrices.
We remark that Hadamard matrix constructions frequently appears in the estimates of Banach-Mazur distances; see e.g.\ the comment on p279 in \cite{ideals}.
\end{rem}
\begin{rem}
Alternatively, one may construct $\mathcal{S}_n$ using the method of \cite[Theorem~1]{barvinok2014thrifty} based on the John
decomposition and a tensor power trick, which actually works when $\mathcal{C}^n$ is a general balanced convex body in $\mathbb{R}^N$;
see discussions around \eqref{e284}. 
\end{rem}

\begin{lem}\label{lem_l1}
Let $\kappa\in(0,\infty)$, $n\in\{1,2,\dots\}$, and $d\in\{0,1,\dots\}$.
Define $h_n$ as in Section~\ref{sec_h}.
Then there exists $c_{\kappa,d}>0$ such that 
for any $s\in(0,1/2)$,
\begin{align}
\limsup_{n\to\infty}h_n(s)\le c_{\kappa,d}s^{\frac1{4}}\log\frac1{s}.
\end{align}
\end{lem}
\begin{proof}
Let $d_1$, $d_2$,\dots be any sequence satisfying $d_i\le d$ for each $i=1,2,\dots$.
Suppose that $\mathcal{A}_n$ is $\kappa$-differentiated,
$\abs{\{z_i\colon \exists z^n\in\mathcal{A}_n\}}\le d+1$, $i=1,\dots,n$, and 
\begin{align}
\mathbb{E}[\rho(\bar{Y}^n)]\le s,
\end{align}
for each $n$, where we recall that $\bar{Y}^n$ is equiprobable on $\mathcal{C}^n$ and
\begin{align}
\rho(y^n):=\log\left(\mathbb{E}[\exp\left< y^n,Z^n\right>]\right)
-\mathbb{E}[\left< y^n,Z^n\right>]
\end{align}
with $Z^n$ equiprobable on $\mathcal{A}_n$.
Now choose 
$\mathcal{S}_n\subseteq\mathcal{C}^n$ as in Lemma~\ref{lem_sampling}
and set $P_{Y^n}$ to be defined by the geometry of $\mathcal{S}_n$ as in \eqref{e_geometry}.
The following defines an action of the multiplicative group $\{-1,1\}^{N}$ on the vertices of $\mathcal{C}^n=[-1,1]^{N}$, which can also be identified as $\{-1,1\}^{N}$:
\begin{align}
u(y^n):=(u_1b_1,\dots,u_{N}b_{N})
\end{align}
where $u=(u_1,\dots,u_{N})$ and $(y_1,\dots,y_n)=(b_1,\dots,b_{N})$ are both in $\{-1,1\}^{N}$.

Now let $U=(U_1$,\dots,$U_{N})$ be equiprobable on $\{-1,1\}^{N}$ and independent of $Y^n\sim P_{Y^n}$.
Observe that 
\begin{align}
\mathbb{E}[\rho(U(Y^n))]
&=\frac1{2^{N}}\sum_{u\in\{-1,1\}^{N}}
\sum_{y^n\in\mathcal{C}^n}P_{Y^n}(y^n)\rho(u(y^n))
\\
&=
\frac1{2^{N}}
\sum_{y^n\in\mathcal{C}^n}
\left(\sum_{u\in\{-1,1\}^{N}}
P_{Y^n}(u(y^n))\rho(y^n)\right)
\label{e_changev}
\\
&=
\frac1{2^{N}}
\sum_{y^n\in\mathcal{C}^n}\rho(y^n)
\label{e_cover}
\\
&=\mathbb{E}[\rho(\bar{Y}^n)]
\end{align}
where
\eqref{e_changev} used change of variables (noting that in the multiplicative group $\{-1,1\}^{N}$, $u$ is the inverse of itself)
and \eqref{e_cover} follows since for any $y^n,z^n\in\mathcal{C}^n$ there is exactly one $u\in\{-1,1\}^{N}$ such that $u(y^n)=z^n$.
Thus we can find a $u$ such that 
\begin{align}
\mathbb{E}[\rho(u(Y^n))]
\le 
\mathbb{E}[\rho(\bar{Y}^n)].
\end{align}
Now define the random variable
\begin{align}
\hat{Y}^n=u(Y^n).
\end{align}
We have
\begin{align}
ns\ge
\mathbb{E}[\rho(\hat{Y}^n)].
\end{align}
But by Markov's inequality
\begin{align}
\mathbb{P}[\left<y^n,Z^n\right>
\ge
\mathbb{E}[\left<y^n,Z^n\right>]
+\lambda]
\le \exp(\rho(y^n)-\lambda),\quad
\forall \lambda\in \mathbb{R},
y^n\in\mathbb{R}^{N}
\end{align}
where $Z^n$ is uniform on $\mathcal{A}_n$.
Now set
\begin{align}
\mathcal{B}_n
:=
\left\{z^n\in\mathcal{A}_n\colon
\left<y^n,z^n\right>
\le 
\mathbb{E}[\left<y^n,Z^n\right>]
+\rho(y^n)+\log(2\abs{\mathcal{S}_n}),
\,
\forall y^n\in u(\mathcal{S}_n)
\right\}.
\end{align}
By applying the union bound to the above we have
\begin{align}
\mathbb{P}[Z^n\in\mathcal{B}_n]
\ge \frac1{2}.
\end{align}
Since $Z^n$ is equiprobable on $\mathcal{A}_n$,
we must have
\begin{align}
\abs{\mathcal{B}_n}\ge \frac1{2}\abs{\mathcal{A}_n}.
\label{e227}
\end{align}

Next we lower bound the packing number of $\mathcal{B}_n$. 
A Hamming ball of diameter $rn$ contains no more than ${n\choose rn}(d+1)^{rn}\le\exp(n[{\rm H}(r)+r\log (d+1)])$ elements. 
Here ${\rm H}(r):=r\log\frac1{r}+(1-r)\log\frac1{1-r}$ denotes the entropy function.
Thus $\mathcal{B}_n$ contains at least 
\begin{align}
\abs{\mathcal{B}_n}\exp(-n[{\rm H}(r)+r\log(d+1)])
\end{align}
elements with pairwise Hamming distance at least $rn$.
Indeed, this follows since the $rn$-blowup of the maximum $rn$-packing under the Hamming distance must contain $\mathcal{B}_n$.
Since $\mathcal{A}_n$, and hence $\mathcal{B}_n$, is $\kappa$-differentiated, we see that a $rn$-packing of $\mathcal{B}_n$ under the Hamming distance is also a $\kappa rn$-packing under $\|\cdot\|_{(\mathcal{C}^n)^{\circ}}$ norm,
and hence $\frac{\kappa rn}{\sqrt{t}}$-packing under $\|\cdot\|_{(\mathcal{S}_n)^{\circ}}$ in view of \eqref{e_scontain}.
Thus
\begin{align}
{\sf P}_{\frac{\kappa rn}{\sqrt{t}}}(\mathcal{B}_n)
\ge \abs{\mathcal{B}_n}\exp(-n[{\rm H}(r)+r\log(d+1)]).
\end{align}
where ${\sf P}_{\frac{\kappa rn}{\sqrt{t}}}(\mathcal{B}_n)$ denotes the $\frac{\kappa rn}{\sqrt{t}}$-packing number of $\mathcal{B}_n$ under $\|\cdot\|_{(\mathcal{S}_n)^{\circ}}$.

On the other hand,
by the definition of $\mathcal{B}_n$ and the fact that $\mathbb{E}[\hat{Y}^n]=0$,
it is clear that
\begin{align}
\mathbb{E}[\sup_{z^n\in\mathcal{B}_n}
\left<\hat{Y}^n, z^n\right>]
&\le
\mathbb{E}[\rho(\hat{Y}^n)]+\log(2\abs{\mathcal{S}_n})
\\
&\le
ns+\log(2\abs{\mathcal{S}_n}).
\end{align}
Invoking Lemma~\ref{lem_general}, we see that the $l$-packing number with respect to the $(\mathcal{S}_n)^{\circ}$ norm satisfies
\begin{align}
{\sf P}_{\frac{\kappa rn}{\sqrt{t}}}
(\mathcal{B}_n)\le 
\left[1+\frac{2\sqrt{t}}{\kappa rn}\,(ns+\log(2\abs{\mathcal{S}_n}))\right]^{N}.
\end{align}
We thus obtain
\begin{align}
\abs{\mathcal{B}_n}\exp(-n[{\rm H}(r)+r\log(d+1)])
\le
\left[1+\frac{2\sqrt{t}}{\kappa rn}\,(ns+\log(2\abs{\mathcal{S}_n}))\right]^{N}.
\end{align}
Taking $\frac1{n}\log$ of both sides and using $\log(1+x)\le x$, $x>0$, we obtain
\begin{align}
\frac1{n}\log\abs{\mathcal{B}_n}-{\rm H}(r)-r\log(d+1)
&\le 
\frac{2d\sqrt{t}}{\kappa rn}\,[ns+\log(2\abs{\mathcal{S}_n})].
\end{align}
By Lemma~\ref{lem_sampling}, 
\begin{align}
\frac1{n}\log\abs{\mathcal{S}_n}
&\le \frac1{n}\log\left((2t)^{\lfloor\frac{N}{t}\rfloor}\cdot
2^{N-\lfloor\frac{N}{t}\rfloor t}
\right)
\\
&\le \frac1{n}\log\left((2t)^{\lfloor\frac{dn}{t}\rfloor}\cdot
2^{dn-\lfloor\frac{dn}{t}\rfloor t}
\right)
\\
&\le \frac1{n}\log\left((2t)^{\frac{dn}{t}}\cdot
2^t
\right)
\\
&= \frac{d}{t}\log(2t)+\frac{t}{n}\log 2.
\end{align}
Let us first suppose that $1/s\in\{2,4,8,\dots\}$ is a power of 2, and fix $t=1/s$ and arbitrary $r>0$ independent of $n$.
Taking $n\to\infty$ we obtain
\begin{align}
r\left(\limsup_{n\to\infty}\frac1{n}\log\abs{\mathcal{B}_n}-{\rm H}(r)-r\log(d+1)\right)
\le
\frac{2d}{\kappa}\sqrt{s}\left(
1+d\log\frac{2}{s}
\right).
\end{align}
Letting $r$ be such that ${\rm H}(r)+r\log(d+1)=\frac1{2}\limsup_{n\to\infty}\frac1{n}\log\abs{\mathcal{B}_n}$, we see that $\frac1{2}\limsup_{n\to\infty}\frac1{n}\log\abs{\mathcal{B}_n}=\Theta(r\log\frac1{r})$,
and hence there exists $c_{\kappa,d}$ large enough such that 
\begin{align}
\limsup_{n\to\infty}\frac1{n}\log\abs{\mathcal{B}_n}
\le c_{\kappa,d}s^{\frac1{4}}\log\frac1{s}
\end{align}
for all $s$ being a reciprocal of a power of $2$.
By choosing $c_{\kappa,d}$ large enough we can extend the bound to all $s\in(0,1/2)$.
Combining with \eqref{e227} we obtain the desired bound on $\frac1{n}\log\abs{\mathcal{A}_n}$ and hence on $h_n(s)$.
\end{proof}

\subsection{Connection to the Banach-Mazur Distance}\label{sec_bm}
Given a bounded, symmetric convex set $\mathcal{C}$ containing the origin in its interior, the norm $\|\cdot\|_{\mathcal{C}}$ defines a Banach space.
Note that applying a nondegenerate linear transform on $\mathcal{C}$ does not change the Banach space structure (i.e., the resulting spaces are equivalent up to isomorphism).
Now suppose that $\mathcal{C}_1$ and $\mathcal{C}_2$ are two such convex sets, and $B_1$ and $B_2$ are the corresponding Banach spaces; equivalently, $\mathcal{C}_1$ and $\mathcal{C}_2$ are the unit balls in the Banach spaces $B_1$ and $B_2$.
The Banach-Mazur distance (see e.g.\ \cite{pajor,ideals}), in this context, is defined as
\begin{align}
d(B_1,B_2):=\inf\{\lambda\ge 1\colon 
\exists A,\, \mathcal{C}_1\subseteq \mathcal{A}\mathcal{C}_2,\,\mathcal{A}\mathcal{C}_2\subseteq \lambda\mathcal{C}_1\}.
\end{align}
Then $\log d(\cdot,\cdot)$ defines a metric on the space of all such Banach spaces.
The Banach-Mazur distance plays an important role in the local theory of Banach spaces initiated by Milman \cite{milman1,milman2}.
In this language, the explicit construction in Lemma~\ref{lem_sampling} essentially established the following result:
there exists a Banach space $B$ where the unit ball has at most $t^{\frac{n}{t}}$ vertices such that 
\begin{align}
d(B,\ell_{\infty}^n)=O(\sqrt{t})
\label{e_bm}
\end{align}
for any $t>2$ independent of $n$, where $\ell_{\infty}^n$ denotes the $n$-dimensional space equipped with the $\ell_{\infty}$ norm.
By duality, this is equivalent to finding Banach space $B$ where the unit ball has at most $t^{\frac{n}{t}}$ \emph{facets} such that 
\begin{align}
d(B,\ell_1^n)=O(\sqrt{t}).
\label{e284}
\end{align}
The result in \eqref{e_bm} is not new; in fact the bound in \eqref{e_bm} holds even if $\ell_{\infty}^n$ is replaced by any $n$-dimensional Banach space \cite{barvinok2014thrifty}.

The technique of this paper can also be used to prove the Gaussian version of the Cover's problem (problem stated and solved in \cite{WuBarnesOzgur}).
In that case, we need an $\ell_2$ version of the estimate in \eqref{e_bm}:
for any $t>2$ independent of $n$,
there exists a Banach space where the unit ball has at most $(1+\frac1{t})^n$ vertices such that  
\begin{align}
d(B,\ell_2^n)=O(\sqrt{t}).
\label{e_bm2}
\end{align}
This is slightly better than the $\ell_1$ case since the required number of vertices is smaller.
The bound in \eqref{e_bm2} can be seen using the construction from \cite[p96-97]{figiel1980large}, and its tightness has been shown in \cite[Theorem~3.3]{pajor}\cite[Corollary~3]{gluskin1989extremal}.

\backmatter

\bmhead{Acknowledgments}
The author gratefully acknowledge Professor Ramon van Handel for very helpful comments on the initial version of the manuscript, especially for pointing out that Lemma~\ref{lem_general} already appeared in the work of Pajor \cite{pajort}.
The author is indebted to Professor Ayfer Ozgur for introducing Cover's problem, discussions during our previous collaborative work \cite{liuozgur}, and her guidance and support as a mentor in the IT society.
The author also thanks Professor Shahar Mendelson for comments on some references about Rademacher complexity.
This work was supported by the start-up grant at the Department of Statistics, University of Illinois.

\bmhead{Data availibility}Data sharing not applicable to this article as no datasets were generated or analyzed during the current study.

%
%
%
%
%
%
%

\begin{appendices}
\section{Proof of Claims in Example~\ref{exp3}}\label{app1}
For any fully supported $P_X$, we can see that $P_{X\vert Z=z}$ (induced by $P_X$ and $P_{Z\vert X}$), $z\in\mathcal{Z}$ are distinct distributions.
Therefore we cannot combine symbols in $\mathcal{Z}$ to form a ``more succinct'' sufficient statistic for $X$. 
However, below we will see that there exists a capacity-achieving $P_X$ which is not fully supported.

Define $P_X=[\frac1{2},\frac1{2},0]$, $P_{YZ\vert X}=P_{Y\vert X}P_{Z\vert X}$, and $Q_{YZ}=\frac1{2}P_{YZ\vert X=1}+\frac1{2}P_{YZ\vert X=2}$.
We will show that 
\begin{align}
D(P_{YZ\vert X=3}\|Q_{YZ})<
D(P_{YZ\vert X=1}\|Q_{YZ})=
D(P_{YZ\vert X=2}\|Q_{YZ})
\label{e100}
\end{align}
for the range of $\epsilon$ and $\delta$ in Example~\ref{exp3}, which will imply that $P_X$ maximizes $I(X;YZ)$ in view of the saddle-point characterization of the channel capacity (Section~\ref{sec_caod}).
To show \eqref{e100}, note that in the matrix form, we have
\begin{align}
Q_{YZ}=
\begin{bmatrix}
\frac1{16}+\frac{\epsilon^2}{4} &
\frac1{16}+\frac{\epsilon^2}{4} &
\frac1{8}-\frac{\epsilon^2}{2}
\\
\frac1{16}+\frac{\epsilon^2}{4} &
\frac1{16}+\frac{\epsilon^2}{4} &
\frac1{8}-\frac{\epsilon^2}{2}
\\
\frac1{8}-\frac{\epsilon^2}{2} &
\frac1{8}-\frac{\epsilon^2}{2} &
\frac1{4} +\epsilon^2
\end{bmatrix}
=
\begin{bmatrix}
\frac1{16} & \frac1{16} & \frac1{8}\\
\frac1{16} & \frac1{16} & \frac1{8}\\
\frac1{8} & \frac1{8} & \frac1{4}
\end{bmatrix}
+\Theta(\epsilon^2)
\label{e101}
\end{align}
where $\Theta(\epsilon^2)$ denotes a matrix whose Frobenius norm is order $\epsilon^2$.
Therefore, we can see (for example, by approximating the relative entropy with the $\chi^2$-divergence, noting $P_{YZ\vert X=1}-Q_{YZ}=\Theta(\epsilon)$) that
\begin{align}
D(P_{YZ\vert X=1}\|Q_{YZ})=D(P_{YZ\vert X=2}\|Q_{YZ})=\Theta(\epsilon^2).
\end{align}
Similarly,
\begin{align}
D(P_{YZ\vert X=3}\|Q_{YZ})=\Theta(\delta^2+\epsilon^4)=\Theta(\epsilon^4).
\end{align}
These establish \eqref{e100} for sufficiently small $\epsilon$, confirming that $P_X$ is capacity-achieving.

To see $R_{\rm crit}={\rm H}(\frac1{2}+2\epsilon^2)$, note that by Definition~\ref{defn1} we have $\underline{1}=\underline{2}\neq \underline{3}$. Therefore 
\begin{align}
Q_{\uY\uZ}=
\begin{bmatrix}
\frac1{4}+\epsilon^2 & \frac1{4}-\epsilon^2\\
\frac1{4}-\epsilon^2 & \frac1{4}+\epsilon^2
\end{bmatrix}
\end{align}
where the first column/row corresponds to $\underline{1}=\underline{2}$ and the last colum/row corresponds to $\underline{3}$.
Therefore $R_{\rm crit}=H(\uZ\vert\uY)={\rm H}(\frac1{2}+2\epsilon^2)$.

From \eqref{e101} we see that 
\begin{align}
H(Z\vert Y)=\frac1{4}H(Z\vert Y=1)+\frac1{4}H(Z\vert Y=2)
+\frac1{2}H(Z\vert Y=3)
=\frac{\log 2}{2}+{\rm H}(\frac1{2}+2\epsilon^2).
\end{align}

\section{Achievability (Proof of Proposition~\ref{prop_achieve})}\label{app_prop_achieve}
Consider encoder, relay, and decoder in Model~1 with the additional restrictions that the channel input symbols $x_1,\dots,x_n$ must be selected from a set $\mathcal{X}_{\rm good}\subseteq\mathcal{X}$ to be defined in \eqref{e_good},
and the relay and the decoder must first preprocess their received vectors $Z^n$ and $Y^n$ to obtain $\uZ^n$ and $\uY^n$ (i.e., applying coordinate-wise maps $Z_i\mapsto \uZ_i$ and $Y_i\mapsto \uY_i$).
The capacity of the restricted model can only be smaller than that of the original Model~1 due to the restrictions, i.e.,
$C_{\rm restrict}(H(\uZ\vert\uY))\le C(H(\uZ\vert\uY))$
Moreover, from the perspectives of the  encoder, relay, and the decoder, the restricted model is essentially also Model~1 but with channels $P_{\uZ\vert X}$ and $P_{\uY\vert X}$ and input alphabet $\mathcal{X}_{\rm good}$.
Using compress-and-forward (see \cite[Proposition~3]{kim_techniques}
with the substitutions $Y_1\leftarrow Z$ and $\hat{Y}_1\leftarrow V$), 
we see that $C_{\rm restrict}(H(\uZ\vert \uY))=C_{\rm restrict}(\infty)$.
It will be shown in Proposition~\ref{prop_delta} that $C_{\rm restrict}(\infty)=C(\infty)$.
Thus $C(H(\uZ\vert \uY))\ge C(\infty)$.

\section{Proof of Theorem~\ref{thm_general}}
\label{app_thm_general}
Most parts of the proofs of upper and lower bounds on $R_{\rm crit}$ follows the same lines as the symmetric case, with the modifications that all underlines for $y$ and $z$ are removed throughout, $\mathcal{X}_{\rm bad}=\emptyset$ in \eqref{e_bad}, and Proposition~\ref{prop7} and Proposition~\ref{prop_delta} are no longer used. Note that when the $z$-equivalence classes are singletons, $E$ in Model~2 is a constant, and in fact Model~2 collapses to Model~1.
The only essential difference in the proof is the argument of the injectivity of $\psi_x$ in the paragraph of \eqref{e_kt}.
We now show the injectivity of $\psi_x$ as the following instead, continuing \eqref{e_kt} (note that now the underlines are removed):
Pick arbitrary $z\neq z'$ in $\mathcal{Z}_x$.
By the second assumption of Theorem~\ref{thm_general} we have $\mathcal{Y}_x=\mathcal{Y}$.
Now by \eqref{e_kt}, if $(\tilde{K}_x(y,z))_{y\in\mathcal{Y}_x}=(\tilde{K}_x(y,z'))_{y\in\mathcal{Y}_x}$ then the vectors $(Q_{YZ}(y,z))_{y\in\mathcal{Y}}$ and $(Q_{YZ}(y,z'))_{y\in\mathcal{Y}}$ differ by a multiplicative constant.
Therefore
$Q_{Y\vert Z}(\cdot\vert Z)=Q_{Y\vert Z}(\cdot\vert Z')$,
which contradicts the assumption that $z$-equivalence classes are singletons, 
Hence $\psi_x(z)=\psi_x(z')$ is false.
Therefore $\psi_x$ is injective.

\section{Proof of Theorem~\ref{thm3}}
\label{app_thm3}
By Theorem~\ref{thm_general}, it suffices to show that for almost all $(P_{Y\vert X}, P_{Z\vert X})$, the $z$-equivalence classes are singletons.

For almost all $(P_{Y\vert X}, P_{Z\vert X})$, we have that 
\begin{align}
P_{Y\vert X=x_1}P_{Z\vert X=x_1}\neq P_{Y\vert X=x_2}P_{Z\vert X=x_2},
\quad \forall x_1, x_2\in\mathcal{X}\colon x_1\neq x_2.
\label{e107}
\end{align}
Under \eqref{e107}, the capacity is positive and the support of any capacity achieving input distribution (CAID) has cardinality at least 2.
We will also show that for almost all $(P_{Y\vert X}, P_{Z\vert X})$,
\begin{align}
\rank([S_{YZ}(\cdot, z_1),S_{YZ}(\cdot,z_2)])\ge 2,\quad
\forall z_1,z_2\in\mathcal{Z}\colon z_1\neq z_2; \,S_X\colon \abs{\supp(S_X)}\ge 2.
\label{e106}
\end{align}
Here, 
$S_{YZ}$ is induced by $S_X$ and $P_{Y\vert X}P_{Z\vert X}$, and
$\rank([S_{YZ}(\cdot, z_1),S_{YZ}(\cdot,z_2)])$ denotes the rank of the matrix $[S_{YZ}(\cdot, z_1),S_{YZ}(\cdot,z_2)]$ where $S_{YZ}(\cdot, z_1)$ and $S_{YZ}(\cdot,z_2)$ are treated as column vectors.
Note \eqref{e107} (hence CAID has support size at least 2) and \eqref{e106} combined imply that the $z$-equivalence classes are singletons, and hence imply the claim of the theorem.

It remains to show \eqref{e106} for almost all $(P_{Y\vert X},P_{Z\vert X})$. 
In turn, it suffices to show that given arbitrary  $\mathcal{S}\subseteq \mathcal{X}$, $\abs{\mathcal{S}}\ge 2$, and any $z_1,z_2\in\mathcal{Z}$, $z_1\neq z_2$,
we have 
\begin{align}
\rank([S_{YZ}(\cdot, z_1),S_{YZ}(\cdot,z_2)])\ge 2,\quad
\forall S_X\colon \supp(S_X)=\mathcal{S}
\label{e109}
\end{align}
for almost all $(P_{Y\vert X}, P_{Z\vert X})$.
To see \eqref{e109},
first call elements in the given $\mathcal{S}$ as $\{x_1,\dots x_s\}$ where $s\ge 2$. Pick arbitrary $\{y_1,\dots,y_s\}\subseteq \mathcal{Y}$.
A submatrix of the matrix $[S_{YZ}(\cdot, z_1),S_{YZ}(\cdot,z_2)]$ equals
\begin{align}
\begin{bmatrix}
P_{Y\vert X}(y_1\vert X_1)\dots P_{Y\vert X}(y_1\vert X_s)\\
\dots\\
P_{Y\vert X}(y_s\vert X_1)\dots P_{Y\vert X}(y_s\vert X_s)
\end{bmatrix}
\begin{bmatrix}
S_X(x_1)&\dots& 0\\
~&\ddots&~\\
0&\dots &S_X(x_s)
\end{bmatrix}
\begin{bmatrix}
P_{Z\vert X}(z_1\vert X_1)\quad P_{Z\vert X}(z_2\vert X_1)\\
\dots\\
P_{Z\vert X}(z_1\vert X_s)\quad P_{Z\vert X}(z_2\vert X_s)
\end{bmatrix}.
\end{align}
For almost all $(P_{Y\vert X},P_{Z\vert X})$, the first matrices above is invertible and the third matrix is rank 2.
The middle matrix is invertible when $\supp(S_X)=\mathcal{S}$. 
Therefore \eqref{e109} holds for almost all $(P_{Y\vert X},P_{Z\vert X})$, as desired.




\end{appendices}


\bibliography{ref_minorization}


\end{document}